\documentclass[conference]{IEEEtran}
\usepackage{amsthm}
\newtheorem{definition}{Definition}
\usepackage{hyperref}
\usepackage{booktabs}
\usepackage{cite}
\usepackage[ruled,vlined,linesnumbered]{algorithm2e}
\usepackage{enumitem}
\usepackage{graphicx}
\usepackage{subcaption}
\usepackage{textcomp}
\usepackage{xcolor}
\usepackage{diagbox}
\usepackage{multirow}
\usepackage{xurl}
\usepackage{comment}
\usepackage{balance}
\usepackage{tcolorbox}

\usepackage[normalem]{ulem}

\newcommand{\bheading}[1]{{\vspace{4pt}\noindent{\textbf{#1}}}}
\newcommand{\iheading}[1]{{\vspace{4pt}\noindent{\textit{#1}}}}

\newcommand{\ubpara}[2][]{{\vspace{1pt}\noindent{\textbf{\uline{#1} #2}}}}

\newcommand{\secref}[1]{\mbox{Sec.~\ref{#1}}\xspace}
\newcommand{\ssecref}[1]{\mbox{\S\ref{#1}}\xspace}

\newcommand{\figref}[1]{\mbox{Fig.~\ref{#1}}}
\newcommand{\tabref}[1]{\mbox{Table~\ref{#1}}}

\newcommand{\algref}[1]{\mbox{Algorithm~\ref{#1}}}
\newcommand{\ignore}[1]{}

\newcounter{note}[section]

\newcommand{\sysname}{\texttt{TeeDAO}\xspace}
\newcommand{\sysnameSGX}{\texttt{TeeDAO-SGX}\xspace}
\newcommand{\sysnameTDX}{\texttt{TeeDAO-TDX}\xspace}
\newcommand{\sysnameCSV}{\texttt{TeeDAO-CSV}\xspace}

\newcounter{packednmbr}

\newenvironment{packeditemize}{
\begin{list}{$\bullet$}{
\setlength{\labelwidth}{0pt}
\setlength{\itemsep}{2pt}
\setlength{\leftmargin}{\labelwidth}
\addtolength{\leftmargin}{\labelsep}
\setlength{\parindent}{0pt}
\setlength{\listparindent}{\parindent}
\setlength{\parsep}{1pt}
\setlength{\topsep}{1pt}}}{\end{list}}

\newtheorem{theorem}{Theorem}
\newtheorem{lemma}{Lemma}
\newtheorem{corollary}{Corollary}

\let\oldnl\nl \newcommand{\nonl}{\renewcommand{\nl}{\let\nl\oldnl}}

\begin{document}
\date{}
\title{\sysname: A Decentralized Autonomous Organization for Heterogeneous TEEs}

\author{\IEEEauthorblockN{Pinshen Xu\IEEEauthorrefmark{1},
Wentao Dong\IEEEauthorrefmark{1},
Guoxing Chen\IEEEauthorrefmark{2}, 
Jianyu Niu\IEEEauthorrefmark{1},
Cong Wang\IEEEauthorrefmark{1}, and
Yinqian Zhang\IEEEauthorrefmark{3}}
\IEEEauthorblockA{\IEEEauthorrefmark{1}City University of Hong Kong}
\IEEEauthorblockA{\IEEEauthorrefmark{2}Shanghai Jiao Tong University}
\IEEEauthorblockA{\IEEEauthorrefmark{3}Southern University of Science and Technology}}

\maketitle
\begin{abstract}
Trusted Execution Environments (TEEs) have emerged as a critical technology for safeguarding sensitive data and ensuring code integrity in modern computing systems. However, relying on a single TEE implementation makes systems vulnerable to a central point of attack. Building distributed-trust systems leveraging heterogeneous TEEs helps disperse trust but still faces threats from centralized management and adaptive mobile adversaries. To address these challenges, this paper introduces \sysname, a novel three-layer framework that automatically organizes multiple heterogeneous TEE instances and provides unified interfaces to support diverse applications, while ensuring long-term guarantees of availability, integrity, and confidentiality. \sysname couples BFT-ordered governance with heterogeneity-aware Distributed Proactive Secret Sharing (DPSS) and Secure Multi-Party Computation (MPC) so that attestation-driven committee changes are consistently reflected in secret recovery, resharing, and computation across a dynamic committee of heterogeneous TEEs.
We implement a prototype of \sysname, integrating COBRA’s DPSS scheme with the HotStuff BFT consensus protocol, and adapt it for Intel SGX, TDX, and Hygon CSV. Evaluations demonstrate that \sysname achieves up to $1.8\times$ higher key-value store throughput in a large cluster with 61 nodes compared to state-of-the-art systems, efficient autonomous management, and minimal computation overhead ($<18\%$) for multi-party computation tasks.

\end{abstract}
 
\section{Introduction} \label{Introduction}
Trusted Execution Environments (TEEs) have emerged as the cornerstone of modern cloud security, enabling the paradigm of Confidential Computing. 
By leveraging hardware-enforced isolation and remote attestation, TEEs allow users to process sensitive data without trusting the underlying operating system or privileged software. 
This hardware-rooted trust model has been widely adopted: processor vendors have rolled out diverse implementations, ranging from process-based Intel SGX~\cite{IntelSoftwareGuard} to VM-based architectures like Intel TDX~\cite{IntelTrustDomain}, ARM CCA~\cite{ltdArmConfidentialCompute}, AMD SEV~\cite{AMDSecureEncrypted}, and Hygon CSV~\cite{DocumentationX86Hygonsecurevirtualizationrst}.
Major cloud providers (e.g., Azure~\cite{ju-shimAzureConfidentialComputing}, AWS~\cite{AWSNitroEnclaves}, Google Cloud~\cite{GoogleConfidentialComputing}) also expose these capabilities as standard managed services.

Despite various proposals using TEEs to protect security-critical services, e.g., confidential analytics~\cite{DBLP:conf/sp/SchusterCFGPMR15, DBLP:conf/sp/PriebeVC18, ohrimenkoObliviousMultiPartyMachine2016}, key management~\cite{DBLP:conf/uss/HerwigGL20, DBLP:conf/ccs/FischVBG17}, and digital-asset custody~\cite{chengEkidenPlatformConfidentialitypreserving2019, lindTeechainSecurePayment2019, zhangTownCrierAuthenticated2016, mateticBITEBitcoinLightweight2019, mecury, TeeRollup},
a single TEE implementation remains questionable in practice~\cite{chenSgxPectreStealingIntel2019, bulckForeshadowExtractingKeys2018, liSystematicLookCiphertext2022, buhrenInsecureProvenUpdated2019, buhrenOneGlitchRule2021, DBLP:conf/uss/LippGSMM16, hahnelHighResolutionSideChannels2017, leeInferringFinegrainedControl2017, leeHackingDarknessReturnoriented2017,  bulckTellingYourSecrets2017, DBLP:conf/ccs/WilkeS024, DBLP:conf/uss/RauscherWW0G25}. 
The major concerns stem from the \textit{central point of attacks}~\cite{emmadautermanReflectionsTrustingDistributed2022}:
Once the underlying hardware trust root is compromised, isolation can fail, and long-lived secrets may be exposed. 
For example, microarchitectural exploits (e.g., transient execution attacks Spectre/Meltdown~\cite{kocherSpectreAttacksExploiting2019, lippMeltdownReadingKernel2018}) and demonstrated key-exfiltration attacks against commodity TEEs~\cite{chenSgxPectreStealingIntel2019, bulckForeshadowExtractingKeys2018, buhrenInsecureProvenUpdated2019, buhrenOneGlitchRule2021} illustrate that such breaks are realistic. Although vendors actively patch these flaws, the history of vulnerabilities suggests an endless arms race, where new zero-day exploits inevitably emerge~\cite{jauernigTrustedExecutionEnvironments2020, liSoKUnderstandingDesign2024}. Beyond breaches of enclave isolation, a single TEE also faces availability threats (e.g., denial-of-service~\cite{woodDenialServiceSensor2002a}) and side-channel attacks~\cite{zhouSideChannelAttacksTen2005}. 
As a result, a single-TEE system can hardly offer durable guarantees of availability, integrity, or confidentiality for security-critical services. 

To address this \textit{central point of attack}, recent research has turned to distributing trust across heterogeneous TEEs. 
A seminal example is SVR3~\cite{connellSecretKeyRecovery2024}, which advances secure key recovery for end-to-end encrypted messaging as deployed in the Signal platform~\cite{SignalappSecureValueRecovery22025}. SVR3 distributes user secret shares across diverse enclave implementations (e.g., SGX, SEV, Nitro) and cloud providers, ensuring that the compromise of a subset of hardware platforms or providers is insufficient to compromise data confidentiality. 
This approach significantly elevates the bar for adversaries to subvert multiple, distinct roots of trust simultaneously. (See more discussion in \secref{bg_TEE}.)

While SVR3 validates the utility of heterogeneous TEEs in static deployments, our work targets a more challenging domain: \emph{long-running} confidential services ranging from databases~\cite{DBLP:conf/sp/SchusterCFGPMR15} and analytics engines~\cite{DBLP:conf/sp/PriebeVC18}, TLS/CDN infrastructure~\cite{DBLP:conf/uss/HerwigGL20}, and blockchain-oriented oracles and exchange applications~\cite{zhangTownCrierAuthenticated2016, DBLP:conf/ccs/BentovJ0BDJ19}. In this setting, security cannot be established once at deployment; it must be maintained over time as the committee reconfigures and as TEEs evolve through disclosure, patching, and revocation. Specifically, the system should automatically remove compromised TEEs and incorporate new ones, all while keeping users' secrets safe. The core difficulty is that these tasks cannot be solved independently. Attestation can indicate which TEE instances remain eligible to serve, but they do not protect secrets that have already been exposed by formerly compromised enclaves. Conversely, DPSS and MPC can refresh and compute over secret shares, but they do not determine which node should receive those shares after disclosure, patching, or replacement. Thus, a long-running heterogeneous TEE service must make trust evidence, committee membership, and secret-share state evolve together. This dynamic environment imposes two distinct requirements compared to the static deployments:

\begin{packeditemize}
    \item \emph{Evidence-driven autonomous management:} Long-running services must repeatedly \emph{reconfigure} in response to maintenance and vulnerability disclosures, and each operational change must be authorized by verifiable evidence and strictly governed by protocol rules. This ensures that the compromise of a single administrative authority does not degenerate into a single point of failure. 

    \item \emph{Proactive secret maintenance:} The system must withstand \textit{adaptive mobile adversaries}~\cite{herzbergProactiveSecretSharing1995, canettiAdaptiveSecurityThreshold1999, vassantlalCOBRADynamicProactive2022}, who may compromise different TEE \textit{implementations}\footnote{We use the term \textit{implementation} to refer to a specific version of a particular type of TEE hardware.} over time. This necessitates mechanisms to periodically refresh secret shares and effectively invalidate the adversary's accumulated knowledge across configurations.
\end{packeditemize}

These challenges further raise the question: \textit{How can we design a framework that ensures confidentiality, integrity, and availability for a long-running distributed system built on heterogeneous TEEs?}

\bheading{Our \sysname Framework.} To answer the question, we propose \textit{TEE \underline{D}ecentralized \underline{A}utonomous \underline{O}rganization} (\sysname), a three-layer framework that manages a committee of heterogeneous TEEs as a distributed-trust infrastructure for long-running confidential services. 
\sysname enforces a single configuration invariant: the same consensus-certified configuration determines which nodes may order requests, hold secret shares, and participate in service.
To this end, \sysname first integrates a rule-based governance layer, inspired by the algorithmic transparency of DAOs~\cite{buterinNextgenerationSmartContract2014, defilippiGovernanceBlockchainSystems2018}, to realize evidence-driven autonomous management. 
It standardizes heterogeneous TEEs into a uniform node abstraction by deriving verifiable reports from remote attestation and authenticated vendor revocation/TCB information. 
These reports enable \sysname to authenticate node rejoin and detect compromised TEEs.
Furthermore, critical actions, such as provisioning nodes, updating configurations, or evicting compromised platforms, require cryptographically verified threshold approval, thereby asserting operational authority against centralized administrators. It replaces administrator-driven control with consensus-governed, rule-based governance that admits, revokes, and reconfigures the committee only via evidence-carrying operations that every member can independently verify.

Second, \sysname stores long-lived secrets and application state as cryptographic shares across heterogeneous TEEs. It uses DPSS to periodically refresh shares as the committee or hardware landscape changes, thereby rendering accumulated leakage obsolete. In addition, \sysname supports secure multi-party computation (MPC) over the distributed secrets.

Except for the above design, \sysname also abstracts all internal complexity through standardized \texttt{read}, \texttt{write}, and \texttt{execute} interfaces, enabling applications to inherit robust distributed trust without redesigning their security foundations.
This unified interface can support a wide array of confidential services, including encrypted key-value stores (KVS)~\cite{bailleuAvocadoSecureInMemory2021, pattukBigsecretSecureData2013, yuanBuildingEncryptedDistributed2016}, cryptocurrency wallets~\cite{mangipudiUncoveringImpactMental2023, yuDontPutAll2024}, and multi-party collaborative analytics~\cite{liagourisSECRECYSecureCollaborative2023, corrigan-gibbsPrioPrivateRobust2017, volgushevDEMOIntegratingMPC2016, zhouShortcutMakingMPCbased2024}. 

We implement a prototype of \sysname within a heterogeneous TEEs environment, integrating COBRA's DPSS scheme~\cite{vassantlalCOBRADynamicProactive2022} atop the HotStuff Byzantine Fault Tolerant (BFT) consensus protocol~\cite{yinHotStuffBFTConsensus2019}. The system is designed to be versatile: we aptly adapt BLS threshold signatures~\cite{bonehShortSignaturesWeil2001} and the general-purpose MP-SPDZ framework~\cite{kellerMPSPDZVersatileFramework2020} to operate over our DPSS layer, deployed across Intel SGX, Intel TDX, and Hygon CSV enclaves. We also conduct performance evaluations of KVS throughput, management cost, and computation overhead introduced by the framework. The evaluation results show that \sysname outperforms state-of-the-art confidential BFT system COBRA in KVS throughput, achieving up to $1.8\times$ higher throughput in the large cluster with 61 nodes, and demonstrates efficient system management with lower reconfiguration and recovery latency in smaller clusters. Additionally, \sysname introduces minimal computation overhead compared to vanilla MPC based on Shamir's secret sharing protocol under the malicious and honest majority models, with performance overhead below 10\% for most tasks and within 18\% for more complex workloads like machine learning.

\bheading{Contributions.} 
This paper makes three contributions: 
\begin{packeditemize}
    \item We identify the lifecycle problem of long-running heterogeneous TEE services: trust evidence, committee membership, and secret-share ownership must evolve consistently under adaptive compromise. We formalize this setting using an \textit{adaptive mobile adversary} and the security goals of safety, liveness, and proactive secrecy.
    
    \item We introduce \sysname, a three-layer framework that provides three unified interfaces to support KVS, cryptocurrency wallets, and multi-party collaborative analytics services. \sysname enforces the consistency by deriving node eligibility from attestation and revocation evidence, committing membership changes through BFT governance, and coupling each configuration transition with DPSS refresh, resharing, recovery, and MPC execution.
    
    \item We deploy \sysname on three mainstream TEE platforms, i.e., Intel SGX, Intel TDX, and Hygon CSV, and conduct extensive evaluations. The results demonstrate that \sysname achieves high KVS throughput, efficient autonomous management, and reasonable MPC computation overhead.
\end{packeditemize}

\section{Background and Motivation} 
\label{Background}

\subsection{Heterogeneous TEE Systems}
\label{bg_TEE}
TEE is a processor extension to protect sensitive code and data even when the operating system (OS) or software administrators are compromised. It provides a secure runtime environment, referred to as an \textit{enclave}\footnote{We use \textit{enclaves} to denote the isolated environments across different TEE technologies.}, which is isolated from other privileged software. 
As said in \secref{Introduction}, a single TEE implementation remains questionable in practice~\cite{chenSgxPectreStealingIntel2019, bulckForeshadowExtractingKeys2018, liSystematicLookCiphertext2022, buhrenInsecureProvenUpdated2019, buhrenOneGlitchRule2021, DBLP:conf/uss/LippGSMM16, hahnelHighResolutionSideChannels2017, leeInferringFinegrainedControl2017, leeHackingDarknessReturnoriented2017,  bulckTellingYourSecrets2017, DBLP:conf/ccs/WilkeS024, DBLP:conf/uss/RauscherWW0G25}. 
To address it, one promising approach is to distribute trust across heterogeneous TEEs, which inherently stems from the following observations.

\iheading{i) Various TEEs are available.} TEEs can be broadly categorized into two types: process-based TEEs, such as Intel SGX~\cite{IntelSoftwareGuard}, and virtual machine (VM)-based TEEs, including Intel TDX~\cite{IntelTrustDomain}, AMD SEV~\cite{AMDSecureEncrypted}, Hygon CSV~\cite{DocumentationX86Hygonsecurevirtualizationrst}, ARM CCA~\cite{ltdArmConfidentialCompute}, IBM PEF~\cite{huntConfidentialComputingOpenPOWER2021}. These mechanisms are widely available on commodity CPUs and exposed as managed confidential computing services by major cloud providers as shown in \figref{fig:tee-timelines}. At the same time, research TEEs such as Keystone~\cite{leeKeystoneOpenFramework2020}, Sanctum~\cite{costanSanctumMinimalHardware2016}, CURE~\cite{bahmaniCURESecurityArchitecture2021}, and TIMBER-V~\cite{weiserTimbervTagisolatedMemory2019} explore alternative enclave abstractions and hardware–software co-designs.
The wide availability of commodity TEEs establishes a practical basis for building heterogeneous TEE systems.

\begin{figure}[t]
  \centering
  \includegraphics[width=0.9\linewidth]{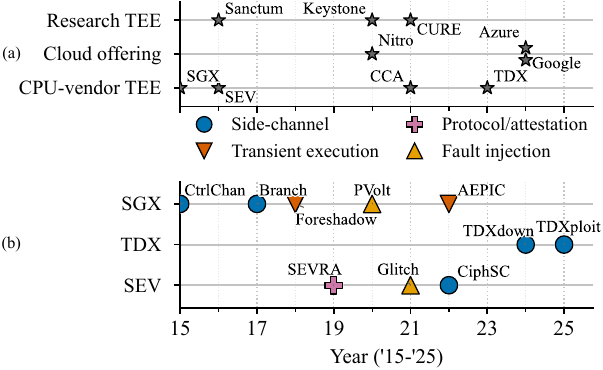}
  \caption{
    (a) TEE availability timeline. (b) Representative attacks over time.
    Controlled-channel attacks (CtrlChan)~\cite{xuControlledChannelAttacksDeterministic2015},
    Branch-shadowing side channels (Branch)~\cite{leeInferringFinegrainedControl2017},
    Foreshadow~\cite{bulckForeshadowExtractingKeys2018},
    Plundervolt (PVolt)~\cite{DBLP:conf/sp/MurdockOGBGP20},
    \AE PIC~\cite{DBLP:conf/uss/BorrelloKSLG022},
    SEV remote-attestation weaknesses (SEVRA)~\cite{buhrenInsecureProvenUpdated2019},
    Fault injection attacks (Glitch)~\cite{buhrenOneGlitchRule2021},
    Ciphertext side channels (CiphSC)~\cite{liSystematicLookCiphertext2022},
    TDXdown~\cite{DBLP:conf/ccs/WilkeS024}, and TDXploit~\cite{DBLP:conf/uss/RauscherWW0G25}. 
    }
  \label{fig:tee-timelines}
  \vspace{-5mm}
\end{figure}

\iheading{ii) Heterogeneous TEEs reduce correlated failures.} Prior studies indicate that different TEEs exhibit distinct vulnerabilities, each tied to their unique design choices and threat models. For instance, as shown in \figref{fig:tee-timelines}, Intel SGX suffers from various software attacks~\cite{checkowayIagoAttacksWhy2013a, leeHackingDarknessReturnoriented2017, weichbrodtAsyncShockExploitingSynchronisation2016}, speculative execution attacks~\cite{chenSgxPectreStealingIntel2019, bulckForeshadowExtractingKeys2018, DBLP:conf/sp/MurdockOGBGP20, DBLP:conf/uss/BorrelloKSLG022}) and side-channel attacks~\cite{shihTSGXEradicatingControlledChannel2017, hahnelHighResolutionSideChannels2017, xuControlledChannelAttacksDeterministic2015, leeInferringFinegrainedControl2017, wangLeakyCauldronDark2017,  bulckTellingYourSecrets2017}), while AMD SEV have been attacked via weaknesses in memory-encryption~\cite{ liSystematicLookCiphertext2022} and integrity trees~\cite{buhrenOneGlitchRule2021, buhrenInsecureProvenUpdated2019}. These differences arise from varying security assumptions, isolation mechanisms, and hardware abstractions, creating unique attack surfaces for each TEE type. Thus, this heterogeneity enables systems to avoid a single point of compromise and achieve stronger security guarantees.

\iheading{iii) Building practical heterogeneous TEE systems is nontrivial.} 
First, if each TEE stores the same plaintext secret, such as users' passwords, compromising any one instance would reveal that secret, even if all other instances remain secure. 
To address this, a promising way is to distribute secret shares among diverse enclave implementations, as illustrated in SVR3~\cite{connellSecretKeyRecovery2024}. This ensures that the compromise of a subset of hardware platforms or providers is insufficient to compromise data confidentiality. (See more in \secref{Introduction}.)

Second, many TEE applications, such as databases~\cite{DBLP:conf/sp/SchusterCFGPMR15} and analytics engines~\cite{DBLP:conf/sp/PriebeVC18}, TLS/CDN infrastructure~\cite{DBLP:conf/uss/HerwigGL20}, and blockchain-oriented applications~\cite{zhangTownCrierAuthenticated2016, DBLP:conf/ccs/BentovJ0BDJ19}, is \emph{long-lived, stateful}. In these deployments, secrets and sensitive state must remain protected \emph{over months/years} despite upgrades, reconfiguration, and especially \emph{adaptive compromise}. 
This requires the system to remove compromised TEEs or authenticate new secure TEEs to join through reconfiguration, while maintaining the security of computation and storage. 

\subsection{Multi-Party Computation (MPC)}
\label{bg_SS}
Secret sharing and multi-party computation are cryptographic techniques that decentralize trust, ensuring that no single party can unilaterally compromise data confidentiality.

\bheading{Threshold Secret Sharing and Proactive Protection.}
Shamir’s threshold Secret Sharing (SS) scheme~\cite{shamirHowShareSecret1979} splits a secret into $n$ shares such that any $t+1$ shares can reconstruct the secret, while $t$ or fewer reveal nothing. Thus, an adversary does not hold the full secret so long as it controls more than $t$ parties. 

In threshold signature schemes (TSS)~\cite{desmedtSocietyGroupOriented1987}, for example, a signing key is shared among $n$ parties, and any $t+1$ parties can jointly produce a valid signature, but no smaller subset can reconstruct the key. Classic SS and TSS, however, are static: the set of parties is fixed, and once a share is compromised, it remains valid indefinitely. DPSS~\cite{vassantlalCOBRADynamicProactive2022, zhouAPSSProactiveSecret2005, schultzMPSSMobileProactive2010, baronCommunicationoptimalProactiveSecret2015, maramCHURPDynamiccommitteeProactive2019} strengthens this by periodically refreshing shares and supporting changes in the committee without reissuing the underlying secret. Modern DPSS constructions (e.g., COBRA~\cite{vassantlalCOBRADynamicProactive2022}) allow parties to collaboratively reshare secrets, recover shares for temporarily offline nodes, and transfer shares to a new committee, thereby providing long-term secrecy even against an \textit{adaptive mobile adversary}. 

\bheading{Secret-Sharing–Based MPC.}
MPC protocols~\cite{yaoProtocolsSecureComputations1982, lindellSecureMultipartyComputation2021} allow several parties to jointly evaluate a function on their private inputs while revealing only the output. In secret-share–based MPC~\cite{escuderoIntroductionSecretSharingBasedSecure2022}, each intermediate value is represented as shares under a linear secret-sharing scheme; additions are performed locally on shares, while multiplications require several interactive communication rounds. General-purpose MPC frameworks such as MP-SPDZ~\cite{kellerMPSPDZVersatileFramework2020} compile high-level programs to arithmetic circuits over shares and support malicious-adversary security. Secret-sharing–based MPC ensures that no single party ever sees plaintext inputs or full application state, while still allowing arbitrary computations on those values.
\section{Problem Statement} \label{Overview}
As motivated in \secref{Background}, relying on a single TEE implementation creates a \emph{central point of attack}:
Once that implementation is broken, the running service within TEEs can lose confidentiality, integrity, or become unavailable.
Prior work~\cite{emmadautermanReflectionsTrustingDistributed2022, connellSecretKeyRecovery2024} has shown that distributing secrets across multiple TEEs can reduce reliance on any single implementation, but supporting \emph{long-running} confidential services still has two challenges, i.e., \textit{evidence-driven autonomous management:} and \emph{proactive secret maintenance}, as presented in \secref{Introduction}. The central consistency requirement is that a node's authority to participate in consensus, store a secret share, and contribute to MPC must be derived from the same certified configuration.
Otherwise, reconfiguration may remove a compromised TEE from the committee while leaving stale shares or computation authority inconsistent with the new membership.

To address these issues, \sysname, a long-running distributed-trust service framework built atop a committee of heterogeneous TEE-equipped servers, is proposed.
It aims to tolerate the compromise of some TEE \emph{implementations} over time and support evidence-driven committee reconfiguration without a trusted human operator. Besides, \sysname is designed to expose a minimal,  uniform API surface for building various confidential services (as discussed in \secref{sec:apps}). 

\begin{figure}[t]
    \centering
    \includegraphics[width=.4\textwidth]{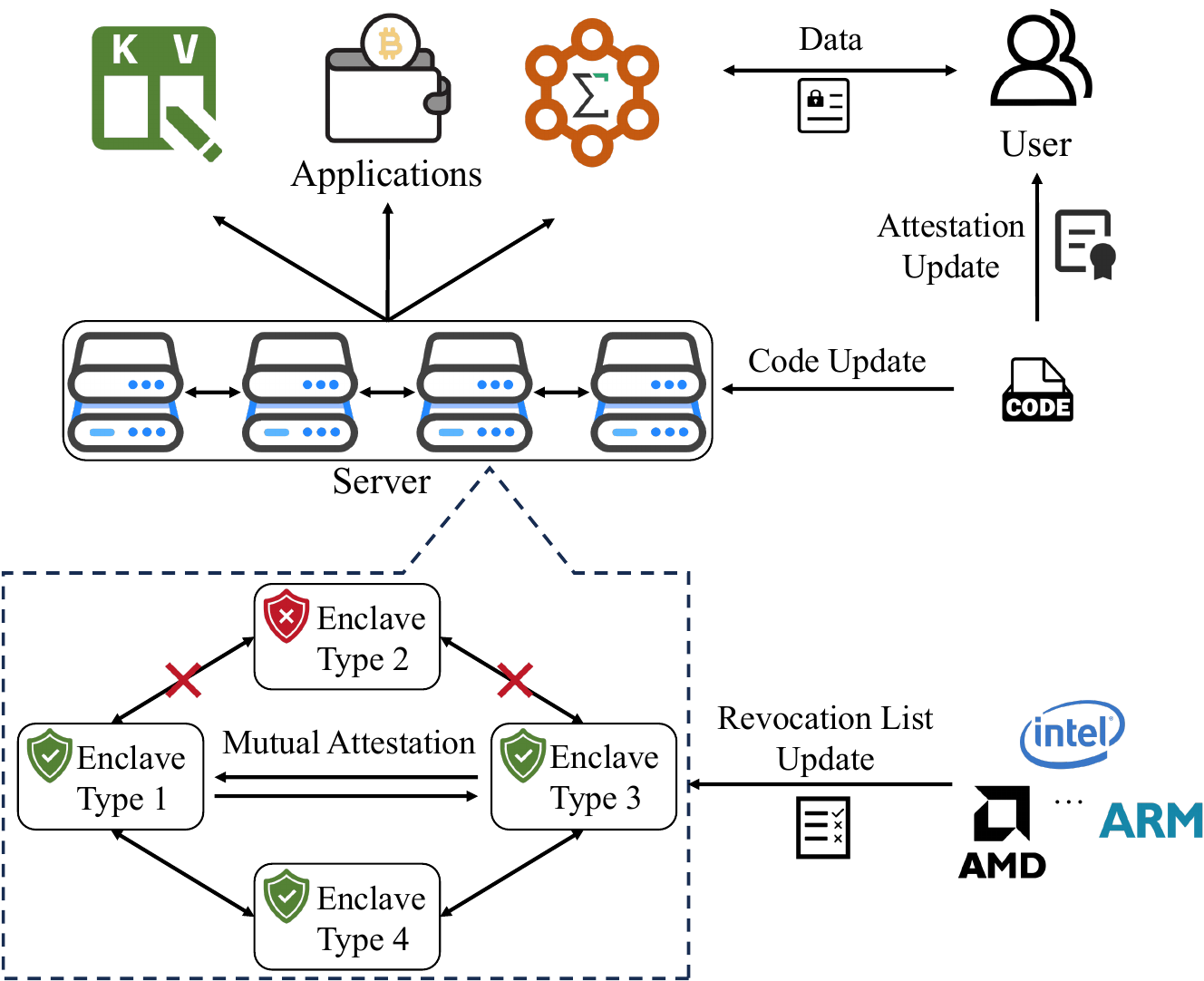}
    \caption{Overview of \sysname services.}
    \label{fig:TEE-DAO_overview}
    \vspace{-6mm}
\end{figure}

\subsection{System Model} \label{sec:sysmodel}
\sysname comprises four entities, as shown in \figref{fig:TEE-DAO_overview}.
\begin{packeditemize}
\item \textbf{TEE committee.}
A committee of $n$ servers runs the \sysname runtime inside enclaves. We use $e_i$ to denote the enclave node hosted on server $i$. Each enclave $e$ has a TEE type $\mathsf{type}(e)$ and an attestation-derived trust status $\mathsf{status}(e)$. The committee maintains replicated service state (including cryptographic secret shares and governance metadata) and collectively executes operations.

\item \textbf{Clients.}
Clients invoke \sysname through its service APIs (e.g., storing/retrieving secrets, and requesting computations) to access developers' applications. Each client request constitutes an \textit{operation}.
Here, application code and configuration updates (e.g., program logic, policies) should remain publicly inspectable/auditable by users or third-party monitors.

\item \textbf{Hardware vendors.}
Vendors (or trusted disclosure channels) publish authenticated vulnerability disclosures and revocation lists $RL$ for TEE implementations.
\sysname consumes $RL$ (together with attestation evidence) to update $\mathsf{status}(e)$ and guide reconfiguration decisions (e.g., removing nodes running revoked versions, allowing patched nodes to rejoin).
\end{packeditemize}

\subsection{Threat Model}
\label{sec:threat}
We consider an \emph{adaptive mobile adversary} $\mathcal{A}$ that controls a dynamic set of corrupted nodes at any given time by aligning with prior work~\cite{vassantlalCOBRADynamicProactive2022}. The size of this set is bounded by the threshold $t$, where $n=3t+1$ is the full committee size. The adversary can also mount network attacks, e.g., intercepting, dropping, delaying, reordering, and replaying messages.

A corrupted node may be affected at different depths. If the adversary controls only the host OS or hypervisor while the TEE remains intact, the node may omit, delay, or withhold protocol messages, but the adversary does \emph{not} learn the enclave's internal state. If the adversary also exploits a vulnerability in the TEE implementation, it obtains the enclave state, including secret shares and signing keys. In this case, the adversary can make the corresponding node deviate arbitrarily from the protocol, including equivocation, sending malformed shares, or returning inconsistent responses. We model such TEE-compromised nodes as Byzantine faults~\cite{castroPracticalByzantineFault2002}. For consensus and service correctness, we conservatively treat all corrupted nodes as faulty. For confidentiality, only TEE-compromised nodes contribute leaked secret shares.

We assume that compromising different TEE hardware implementations requires distinct exploits.
Due to the heterogeneous designs, we assume that the adversary cannot compromise more than $t$ TEEs at any time.
This allows the system to proactively refresh secrets before the security threshold is breached.
We also assume TEE vulnerabilities eventually become known via authenticated vendor channels (e.g., revocation lists and attestation~\cite{TrustedComputingBase, AffectedProcessorsTransient, AMDSEVConfidential}).
TEE nodes whose attestation status becomes invalid are treated as \emph{revoked} and removed.
Additionally, if cryptographic evidence proves a \emph{safety} violation (e.g., equivocation) by a node, it is immediately treated as Byzantine and removed.

\subsection{Security Goals}
\label{sec:goals}

We formalize \sysname's security through three properties: (i) \emph{Safety} (Agreement + Linearizability), capturing the consistency of the replicated service; (ii) \emph{Liveness}, ensuring service availability; and (iii) \emph{Proactive Secrecy}, capturing confidentiality against adaptive adversaries.

\begin{definition}[Agreement]
\label{def:agreement}
\sysname satisfies \emph{Agreement} if, for every configuration $\mathsf{cfg}$, all correct committee members commit the same totally ordered log $\mathsf{Log}_{\mathsf{cfg}}$.
\end{definition}

\begin{definition}[Linearizability]
\label{def:linearizability}
\sysname satisfies \emph{Linearizability} if the execution history of client operations is equivalent to a legal sequential history that preserves the real-time ordering of operations, consistent with the sequence of committed logs $\mathsf{Log}_{\mathsf{cfg}_0},\mathsf{Log}_{\mathsf{cfg}_1},\ldots$.
\end{definition}

\begin{definition}[Liveness]
\label{def:liveness}
\sysname satisfies \emph{Liveness} if, assuming a partially synchronous network and a correct client, any valid operation submitted by the client is eventually committed, executed, and terminated.
\end{definition}

\begin{definition}[Proactive Secrecy]
\label{def:secrecy}
\sysname satisfies \emph{Proactive Secrecy} if the adversary's view across multiple epochs is computationally indistinguishable from a simulation that does not possess the secret inputs. Specifically, previously exposed shares in configuration $\mathsf{cfg}_i$ do not reduce the entropy of secrets refreshed in configuration $\mathsf{cfg}_{j}$ (where $j > i$).
\end{definition}

\subsection{Challenges and Solutions}
\label{sec:challenges}
Building a long-running distributed-trust system under the \textit{adaptive mobile adversary} model presents unique challenges. We identify three primary challenges in eliminating centralized control, ensuring durability against cumulative compromise, and managing heterogeneity, alongside our proposed solutions. These challenges correspond to one lifecycle: evidence changes trust status, trust status changes membership, and membership changes where secret shares and MPC authority reside.

\ubpara[Challenge-1:]{Removing central points of attack.}\label{Challenge_1} 
In traditional distributed systems, critical operations, such as committee admission, revocation, and recovery, commonly rely on privileged administrator credentials.
If these credentials are compromised, the adversary can manipulate committee composition and reduce the security of the distributed system back to \textit{central point of attack}.

\ubpara[Solution-1:]{Consensus-governed committee management using verifiable evidence.}\label{Solution_1}  \sysname replaces administrator-driven committee management with committee decisions agreed by the BFT consensus protocol. Committee-management actions (e.g., admission, revocation, and recovery) are expressed as consensus-ordered operations and are triggered by verifiable evidence, including authenticated vendor revocation lists/attestation evidence and protocol-level misbehavior evidence.

\ubpara[Challenge-2:]{Guaranteeing adaptive, long-term security.}\label{Challenge_2} 
Heterogeneity alone is insufficient for long-term security. 
Given enough time, distinct TEE implementations will inevitably suffer from the accumulated zero-day vulnerabilities. As described in our threat model, an \textit{adaptive mobile adversary}~\cite{herzbergProactiveSecretSharing1995, canettiAdaptiveSecurityThreshold1999, vassantlalCOBRADynamicProactive2022} can sequentially compromise different nodes over time. 
Without a mechanism to invalidate past breaches, the adversary can accumulate secret shares across different configurations until the threshold is reached, eventually recovering the full secret.

\ubpara[Solution-2:]{Proactive share refresh and recovery tied to reconfiguration.}\label{Solution_2}  \sysname binds secret maintenance to committee reconfiguration. When nodes are removed (e.g., due to revocation or detected compromise), \sysname refreshes/re-shares secrets among the remaining members so that previously exposed shares become useless. When recovered/patched nodes are admitted, \sysname establishes fresh shares consistent with the new membership, mitigating long-term share accumulation.

\ubpara[Challeng-3:]{Supporting diverse confidential services.}\label{Challenge_3} 
To facilitate the development of distributed trust systems that support various applications, including encrypted key-value store~\cite{bailleuAvocadoSecureInMemory2021, pattukBigsecretSecureData2013, yuanBuildingEncryptedDistributed2016}, cryptocurrency wallets~\cite{mangipudiUncoveringImpactMental2023, yuDontPutAll2024}, multi-party collaborative analytics~\cite{liagourisSECRECYSecureCollaborative2023, corrigan-gibbsPrioPrivateRobust2017, volgushevDEMOIntegratingMPC2016, zhouShortcutMakingMPCbased2024}, it is also challenging to design unified interfaces capable of accommodating diverse storage and computational operations on long-term distributed secrets. 

\ubpara[Solution-3:]{Unified service APIs.}\label{Solution_3} 
\sysname provides a consistent abstraction layer with minimal APIs
(e.g., \textsf{write}/\textsf{read}/\textsf{execute} interfaces) that abstracts away heterogeneity, reconfiguration, and protocol details, enabling multiple services to share the same distributed-trust framework.

\subsection{Primitives and Building Blocks}\label{sec:primitives}
This section introduces the primitives and building blocks used. We also list the corresponding notations in Table~\ref{tab:notation}. 

\bheading{Committee and Configurations.}
We assume a \emph{partially synchronous} network model, where messages between honest nodes are delivered within a bounded delay after an unknown Global Stabilization Time (GST). This assumption is widely used in BFT protocols~\cite{castroPracticalByzantineFault2002, yinHotStuffBFTConsensus2019, Ladon2025, fast-hotstuff}. Nodes communicate via authenticated point-to-point channels (e.g., TLS).

The system evolves through a sequence of configurations, denoted as $\mathsf{cfg}_0,\mathsf{cfg}_1,\dots$.
For each configuration $\mathsf{cfg}$,  \sysname maintains a common governance state $\sigma_{\mathsf{cfg}}$ that includes the membership, governance rules, and the approved code digest $h_{\mathsf{cfg}}$ needed to validate and execute actions.
Committee management is expressed as governance proposals $p$. A proposal is admissible in configuration $\mathsf{cfg}$ only if $\mathsf{isValid}(p,\sigma_{\mathsf{cfg}})$ holds, including checks driven by evidence (e.g., $RL$, attestation), policy, and the heterogeneity constraint $\mathsf{Hetero}(\mathsf{cfg})$.

\begin{table}[t]
\centering
\footnotesize
\setlength{\abovecaptionskip}{2pt}
\renewcommand{\arraystretch}{0.95}
\setlength{\tabcolsep}{4pt}

\caption{Summary of notations.}
\label{tab:notation}

\begin{tabular*}{\columnwidth}{@{\extracolsep{\fill}}
m{1.6cm}p{3.0cm}|
m{1.3cm}p{2.1cm}@{}}
\toprule
\textbf{Term} & \textbf{Description} & \textbf{Term} & \textbf{Description} \\
\midrule

$n$ & Committee size &
$\mathsf{Log}_{\mathsf{cfg}}$ & Committed log \\

$t$ & Compromise threshold &
$\mathsf{Cert}_{\mathsf{cfg}}(x)$ & Commit certificate \\

$e_i$ & Committee enclave &
$\mathsf{Hetero}(\mathsf{cfg})$ & Type-count bound \\

$\mathsf{isValid}()$ & Admissibility predicate &
$\mathsf{type}(e)$ & TEE type \\

$\mathsf{status}(e)$ & Trust status &
$\sigma_{\mathsf{cfg}}$ & Governance state \\

$RL$ & Revocation list &
$h_{\mathsf{cfg}}$ & Code digest \\

$p$ & Governance proposal &
$\mathsf{cfg}_i$ & Configuration \\

\bottomrule
\end{tabular*}

\vspace{-3mm}
\end{table}

\bheading{Ordered Operations and APIs.}
Both client requests and governance actions are represented as operations and are totally ordered by a BFT replication protocol.
Within each configuration $\mathsf{cfg}$, committee members agree on the committed operation log $\mathsf{Log}_{\mathsf{cfg}}$ and execute operations in that order.
We write $\mathsf{Cert}_{\mathsf{cfg}}(x)$ for a commit certificate that operation is committed under $\mathsf{cfg}$ (e.g., a QC or threshold signature),
which is used to justify committed operations and configuration transitions.
Each API invocation returns a return code $rc$ indicating success/failure of the request.

\bheading{Building Blocks and Their Assumptions.}
We instantiate \sysname using established cryptographic and distributed building blocks. 
Specifically, we employ HotStuff~\cite{yinHotStuffBFTConsensus2019} for consensus, standard ECDSA signature for certificates, Shamir's Secret Sharing (SSS)~\cite{shamirHowShareSecret1979} augmented with Feldman's VSS~\cite{feldmanPracticalSchemeNoninteractive1987} for verifiable secret sharing, and MP-SPDZ~\cite{kellerMPSPDZVersatileFramework2020} for MPC computation. 
Our analysis relies on their standard properties:

\begin{packeditemize}
  \item[(A1)] \textit{BFT safety and liveness.} 
  The underlying consensus protocol (e.g., HotStuff) guarantees that correct nodes commit a unique, totally ordered log (i.e., safety) and make progress during periods of synchrony (i.e., liveness), provided $n \ge 3t+1$.

  \item[(A2)] \textit{Certificate unforgeability.}
  The signature scheme is unforgeable under chosen-message attacks (EU-CMA). 
  The adversary cannot forge a valid commit certificate $\mathsf{Cert}_{\mathsf{cfg}}(x)$ without the secret keys of a quorum of committee members.

  \item[(A3)] \textit{DPSS correctness and secrecy.} 
  The secret sharing scheme ensures that: (i) any set of $t+1$ valid shares can reconstruct the secret; (ii) any set of $\le t$ shares reveals no information about the secret; and (iii) shares refreshed via proactive resharing are independent of previous shares.

  \item[(A4)] \textit{MPC correctness and security.}
  MPC protocols (e.g., MP-SPDZ) provides privacy and correctness against up to $t$ malicious parties. 

  \item[(A5)] \textit{Heterogeneous independence.}
Vulnerabilities in distinct TEE implementations are uncorrelated. 
\end{packeditemize}

\section{\sysname Framework} \label{Framework}

\begin{figure}[t]
    \centering
    \includegraphics[width=.4\textwidth]{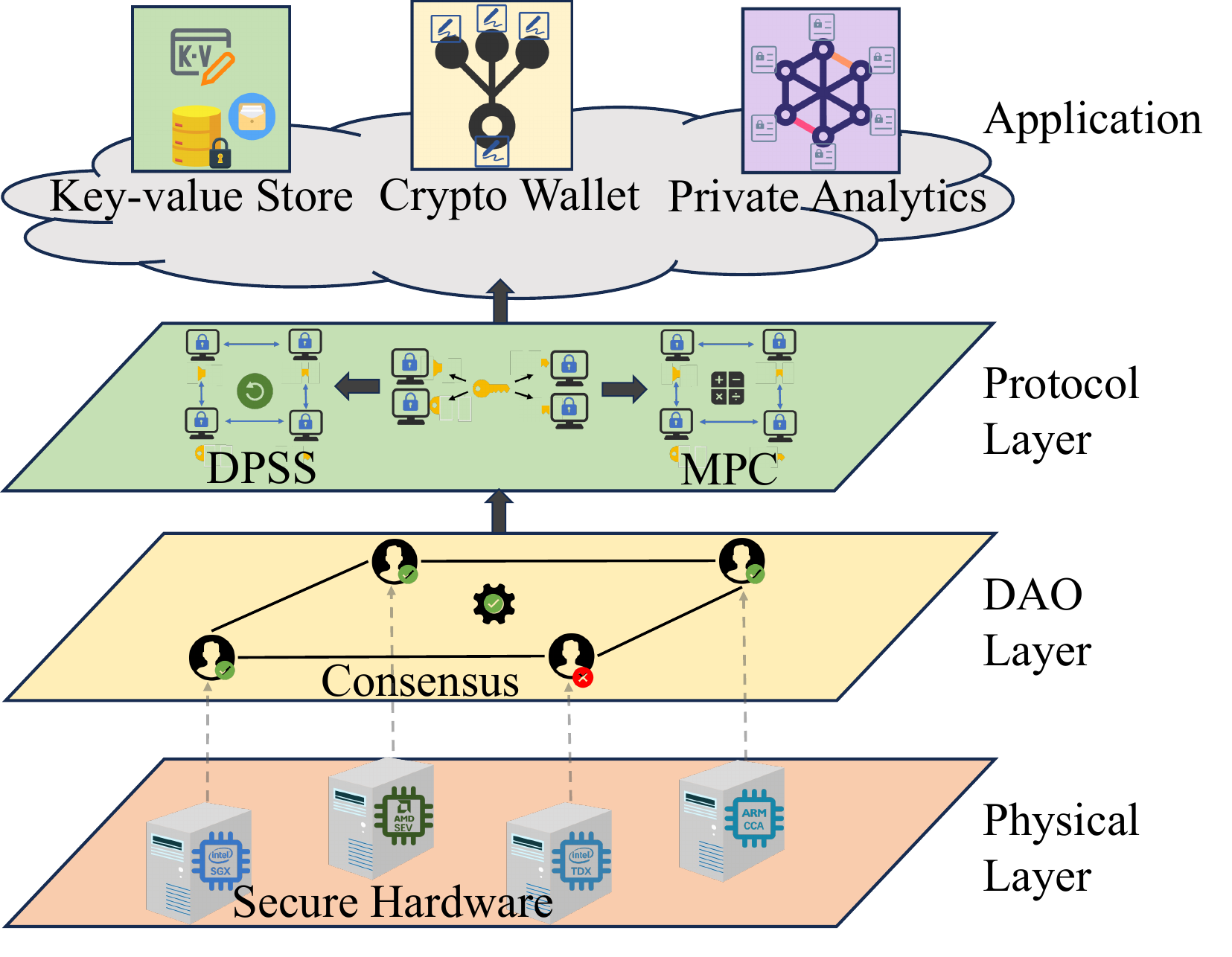}
    \caption{A three layer architecture of \sysname.}
    \label{fig:tee-dao_framework}
    \vspace{-6mm}
\end{figure}

\subsection{Architectural Overview}
\label{sec:arch}
\figref{fig:tee-dao_framework} gives an overview of the \sysname framework. \sysname follows a \emph{three-layer} design consisting of the \emph{Physical}, \emph{DAO}, and \emph{Protocol} layers. The layers separate hardware-rooted trust establishment, consensus-governed committee evolution, and confidential services supporting. This separation allows \sysname to use heterogeneous TEEs as committee members while keeping membership changes auditable and binding protocol state maintenance to the resulting configuration sequence. We briefly introduce each component below and describe them in detail in the following subsections.

\begin{packeditemize}
\item \textbf{Physical Layer (\ssecref{sec:physical}).} 
The Physical layer provides a uniform abstraction for committee members running on heterogeneous TEEs. It derives each member's trust status from vendor attestation and authenticated vulnerability or revocation information. These trust signals serve as evidence for higher-layer governance decisions in heterogeneous deployments (Solution-1).

\item \textbf{DAO Layer (\ssecref{sec:dao}).}
The DAO layer replaces administrator-controlled committee management with \emph{consensus-governed} management (Solution-1). Admission, revocation, quarantine, and recovery are expressed as ordered governance operations, producing an agreed configuration sequence. Using the evidence model from \secref{sec:threat}, the DAO layer can quarantine suspected nodes and remove \emph{revoked} nodes according to policy, making membership and authorization changes auditable and collectively agreed.

\item \textbf{Protocol Layer (\ssecref{sec:protocol}).}
The Protocol layer runs storage and computation protocols over the DAO-managed committee. It exposes a minimum, consistent API (\texttt{write/read/execute} interfaces) for services such as KVS, cryptocurrency wallets, and private analytics (Solution-3). To preserve long-term security against a mobile adversary, protocol state maintenance is tied to reconfiguration: when the DAO layer changes the committee (e.g., due to revocation or admission), the Protocol layer performs the corresponding refresh, resharing, or recovery procedure (Solution-2).
\end{packeditemize}

\subsection{Physical Layer}
\label{sec:physical}
The Physical layer establishes hardware-rooted identities for heterogeneous TEE nodes and produces attestation reports for the upper layers. It exports a uniform node abstraction in which each node identity is bound to the authorized \sysname runtime and to a vendor-authenticated platform security version.

\bheading{Interfaces to DAO Layer.} For each enclave $e$, the Physical layer exports a type label $\mathsf{type}(e)$ identifying its trust domain and an attestation-derived trust status $\mathsf{status}(e)$. \sysname uses three trust statuses. A \textsf{Trusted} enclave presents fresh, verifiable attestation evidence that binds an enclave-generated public key to the approved code digest and to a non-revoked platform version. A \textsf{Revoked} enclave is associated with evidence of an unauthorized code digest or a vulnerable platform version. An \textsf{Expired} enclave cannot provide sufficiently fresh and verifiable evidence, for example, because the relevant vendor-authenticated security information is missing or stale. These signals feed the DAO-layer admissibility predicate $\mathsf{isValid}(\cdot,\sigma_{\mathsf{cfg}})$ and drive membership management.

\bheading{TEE Abstraction.} We treat heterogeneous TEEs as a uniform abstraction as long as they provide the following primitives:

\begin{packeditemize}
    \item \textit{Attestation.} An attestation report binds the enclave measurement, the reported Trusted Computing Base (TCB) version, and an enclave-generated public key used to authenticate subsequent protocol messages.
    
    \item \textit{Isolation.} An uncompromised enclave can protect its runtime state from a compromised host OS or hypervisor.
    
    \item \textit{Revocation Lists.} Each trust domain publishes authenticated information identifying revoked or vulnerable platform versions (revocation lists, TCB information, or equivalent). We denote this input as $RL$.
\end{packeditemize}
To verify attestation reports and $RL$, enclaves pin vendor verification material (e.g., root certificates or public keys) as part of the certified configuration metadata (carried in $\sigma_{\mathsf{cfg}}$).

\bheading{Mutual Attestation.}
Given the attestation abstraction described above, the system must ensure that only uncompromised TEEs can participate in the protocol. Mutual attestation~\cite{chenMAGEMutualAttestation2022} provides this guarantee: two enclaves verify each other's identity by checking their measurements and software versions, thereby excluding enclaves running untrusted or tampered code.

\iheading{Leader-based trust-chain attestation.} A naïve deployment would require every pair of the $n$ enclaves to attest to each other, resulting in $O(n^2)$ attestations during initialization. To make the process scalable, we adopt a leader-based trust-chain attestation while keeping verification fully decentralized:

\begin{packeditemize}
\item \textit{Report generation.} When a node joins or refreshes its status, it generates an ephemeral public key inside the enclave and produces an attestation report that binds this key to the \sysname runtime digest, the reported TCB version, its TEE type, and a freshness proof such as a timestamp or nonce.

\item \textit{Verification against configuration policy.} The designated leader enclave and members can mutually verify each other's report using the pinned vendor roots and the latest available $RL$ for the corresponding enclave. A report is accepted only when the attested code digest matches the governance-approved digest $h_{\mathsf{cfg}}$, the reported security version is not revoked by the relevant $RL$, and the freshness proof falls within the accepted window.
\end{packeditemize}

Once trust is established, the leader aggregates all attestation results and broadcasts a system-wide attestation report. This evidence is incorporated into governance proposals and becomes a committee membership decision through the DAO layer's ordered execution and admissibility checks.

A compromised leader cannot unilaterally admit an untrusted node. Admission proposals must carry the attestation report for each candidate enclave and the digest of the referenced vendor revocation list. Every member can therefore re-verify the report signature against pinned vendor roots and local policy. If the leader equivocates, omits required evidence, or proposes a node that fails verification, correct members reject the proposal; the protocol then times out and selects a new leader.

\bheading{Timely Status Refresh.}
Attestation does not detect unknown zero-day vulnerabilities. Instead, \sysname targets timely reaction once vendor-authenticated security information becomes available.
We adopt two optimizations for a timely refresh:

\iheading{i) Periodic refresh.} Each enclave periodically fetches the latest $RL$ (through the untrusted host but verified inside the enclave) and re-evaluates active peers’ reported security versions.  

\iheading{ii) Event-driven refresh.} Any enclave that observes a newer vendor-authenticated $RL$ than the one referenced by the current configuration may submit evidence to the leader when that update marks an active member as \textsf{Revoked}. The leader then broadcasts a governance proposal so that the committee can adopt the update consistently. If an enclave's status cannot be refreshed within a bounded staleness window, the enclave marks $\mathsf{status}(e)=\textsf{Expired}$, allowing DAO policy to conservatively quarantine nodes whose trust cannot be validated.

\subsection{DAO Layer}
\label{sec:dao}
The DAO layer provides consensus-governed committee management. It replaces administrator-driven control with an auditable, rule-based workflow that turns TEE trust signals and protocol-level evidence of misbehavior into consensus-backed membership and policy decisions. These decisions advance the system through a configuration sequence $\mathsf{cfg}_0,\mathsf{cfg}_1,\ldots$ and expose an externally verifiable record of governance actions.

\bheading{Interfaces to Protocol Layer.} 
The DAO layer exports the current configuration $\mathsf{cfg}$, commit certificates for governance actions and configuration transitions, and the replicated governance state $\sigma_{\mathsf{cfg}}$. The configuration identifies the current member set and certified metadata. A certificate $\mathsf{Cert}_{\mathsf{cfg}}(\cdot)$ attests that an action or transition was committed in the ordered log $\mathsf{Log}_{\mathsf{cfg}}$. The governance state records the active rule set, membership metadata, group verification key $\mathsf{MPK}$, approved code digest $h_{\mathsf{cfg}}$, and current phase. The Protocol layer treats a governance action as effective only when it is accompanied by a valid certificate and runs configuration-triggered refresh or recovery procedures only for committed transitions.

\begin{algorithm}[t]
\DontPrintSemicolon 
\caption{Committee management for compromised and recovered members}
\label{IdentifyAndManagement}

\SetNlSty{}{}{.} 
\SetAlgoNlRelativeSize{0} 
\SetNlSkip{1em} 

\nonl\textbf{Step 1: Discover suspect nodes}\;
    \nonl\textit{Method 1. status update}\;
        Leader compares $\mathsf{type}(e)$ of active members with $RL$.\;
        Adds nodes in $RL$ to $\mathsf{list_{susp}}$.\;
    \nonl\textit{Method 2. rules violation}\;
        \nonl\ForEach{$e_i$}{
            Submit evidence $\mathsf{evid}_i = \{ \mathsf{msg}_i, \mathsf{sig}_i\}$ and violated rules $r$ in accusation $\mathsf{acc}_i = \{i, \mathsf{evid}_i, r \}$ to leader.\;
            Leader assesses accusation, verifies message authenticity and rule violation using $verify(\mathsf{msg}_i, \mathsf{sig}_i)$ and $verify(\mathsf{msg}_i, r)$.\;
            \nonl\If{verification is true}{
                Add accused node to $\mathsf{list_{susp}}$.\;
            }
            \nonl\Else{
                Add accusing member to $\mathsf{list_{susp}}$.\;
            }
        }
\nonl\textbf{Step 2: propose removal of suspect nodes}\;
Leader broadcasts a removal proposal $p$ containing $\mathsf{list}_{\mathsf{susp}}$ and supporting evidence.\;

\nonl\textbf{Step 3: admit patched nodes}\;
Leader attests the requesting node and broadcasts an admission proposal $p$.\;
    \nonl\ForEach{$e_i$}{
            Attest the requesting node and vote for $p$ only if the attestation report is valid 
        }
\end{algorithm}

\bheading{Governance State and Proposals.}
In each configuration $\mathsf{cfg}$, the committee maintains a deterministic governance state $\sigma_{\mathsf{cfg}}$. Committee-management actions are expressed as proposals $p$, covering actions such as removal, quarantine, and admission. A proposal is admissible only when $\mathsf{isValid}(p,\sigma_{\mathsf{cfg}})$ holds. This predicate captures the DAO layer's rule-based policy by checking authorization, attached evidence, and the heterogeneity constraint $\mathsf{Hetero}(\mathsf{cfg})$.

\bheading{Evidence and Rule-based Enforcement.}
We next describe the inputs for $\mathsf{isValid}(\cdot)$ and rule-based detection. The DAO layer consumes two classes of evidence:
\begin{packeditemize}
    \item Physical layer trust signals. Namely, each node's type label $\mathsf{type}(e)$, its attestation-derived status$\mathsf{status}(e)$, and authenticated vendor security updates such as $RL$. These inputs determine whether a node remains trusted, has become revoked, or should be treated as expired.
    
    \item Protocol-level misbehavior evidence. When a member observes misbehavior, it can submit a signed evidence bundle $\mathsf{evid}_i=\{\mathsf{msg}_i,\mathsf{sig}_i\}$ and an accusation $\mathsf{acc}_i=\{i,\mathsf{evid}_i,r\}$ identifying the violated rule $r$ in $\sigma_{\mathsf{cfg}}$.
\end{packeditemize}
Members process these inputs according to governance policy and update the suspect set $\mathsf{list}_{\mathsf{susp}}$.

\bheading{Consensus-ordered Governance.}
Evidence becomes effective only after it is turned into a committee decision. All governance actions are executed via the same consensus-ordered workflow: proposals are broadcast, totally ordered, and applied identically by correct nodes to advance from $\mathsf{cfg}$ to $\mathsf{cfg}'$. \algref{IdentifyAndManagement} summarizes this workflow for suspect discovery, removal, and admission:
\begin{packeditemize}
    \item \emph{Discover suspects.} The leader collects and validates evidence (status updates and accusations) and derives a candidate management action (e.g., remove or admit).
    
    \item \emph{Propose an action.} The leader packages the intended action as a proposal $p$ with the evidence required for $\mathsf{isValid}(p,\sigma_{\mathsf{cfg}})$.
    
    \item \emph{Order and commit.} The committee orders $p$ in $Log_{\mathsf{cfg}}$, producing $Cert_{\mathsf{cfg}}(p)$, after which all correct members apply $p$ identically to obtain $\sigma_{\mathsf{cfg}'}$ and a new configuration $\mathsf{cfg}'$.
\end{packeditemize}
Because every effective governance change must be both admissible and committed, neither a member nor a leader can unilaterally change membership or policy.

\bheading{Lifecycle Phases.}
The DAO layer also exposes a phase flag $\mathsf{phase}\in\{\mathsf{Normal},\mathsf{Reconfig}\}$ as a coordination signal. During $\mathsf{Normal}$, the system provides the full service under the current configuration. When a committed governance decision installs $\mathsf{cfg}'$(e.g., after removal or admission), the system enters $\mathsf{Reconfig}$. The Protocol layer then executes the refresh, resharing, or recovery procedure required by the new configuration. Once those procedures are complete under $\mathsf{cfg}'$, the system returns to $\mathsf{Normal}$.

\subsection{Protocol Layer}
\label{sec:protocol}
The Protocol layer turns a DAO-installed configuration $\mathsf{cfg}$ into an application-facing confidential service. It stores secrets using DPSS and executes computations over those secrets using MPC. Every protocol outcome is configuration-dependent: requests are ordered in $\mathsf{Log}_{\mathsf{cfg}}$ and returned with a certificate verifiable under the group verification key $\mathsf{MPK}$. This design maintains security across heterogeneous TEEs by binding membership and share movement to attestation-derived trust signals and $\mathsf{Hetero}(\mathsf{cfg})$.

\bheading{Interfaces to the Application.}
This layer exposes three application interfaces. A client calls $\textsf{write}(k,v,md)\rightarrow rc$ to store or update a secret value $v$ under key $k$ with metadata $md$. It calls $\textsf{read}(k)\rightarrow (v,md,rc)$ to retrieve the value and metadata associated with $k$. It calls $\textsf{execute}(\textit{code},\textit{args})\rightarrow(\textit{out},rc)$ to run an approved computation over secret-shared state.

The return code $rc$ reports the protocol-level outcome and is distinct from the attestation-derived trust status of an enclave. \textsf{SUCCESS} means that the request was committed and certified. \textsf{NOTFOUND} means that the requested key is absent. \textsf{RECONFIG} indicates that the service is temporarily unavailable because the committee is reconfiguring. \textsf{ABORT} indicates that an inconsistency or malicious behavior was detected during protocol execution. For every \textsf{SUCCESS} response, the client obtains a certificate $\mathsf{Cert}_{\mathsf{cfg}}(\cdot)$ binding the result to $\mathsf{Log}_{\mathsf{cfg}}$. The client accepts the response only if the certificate verifies under $\mathsf{MPK}$ and configuration.

\bheading{Protocol State and Authenticated Membership.}
Each member enclave $e_i$ maintains a protected state $\gamma_i$. This state contains a key-value table whose per-key records have the form $\langle k, v_i, c_i\rangle$, where $v_i$ is a Shamir share of $v$ and $c_i$ is a commitment used to detect malformed shares. It also contains the current protocol phase specified by $\sigma_{\mathsf{cfg}}$. All protocol messages are carried over mutually authenticated channels established through Physical-layer attestation. Each received message is therefore attributable to a member identity in $\mathsf{cfg}$, together with its recorded TEE type and trust status.

\bheading{Normal-phase Operations.}
\sysname stores each secret value $v$ as a degree-$t$ Shamir sharing across the $n$ members in $\mathsf{cfg}$. During the $\mathsf{Normal}$ phase, the committee first orders each request in $\mathsf{Log}_{\mathsf{cfg}}$ and then executes the corresponding DPSS or MPC procedure inside enclaves. The protocol returns \textsf{SUCCESS} only after the request is committed and certified. As shown in \figref{fig:TEE-DAO_service}, \sysname supports three operations.

\begin{packeditemize}
    \item For \textsf{write}, members store their received record $\langle k,v_i,c_i\rangle$ in $\gamma_i$, and the protocol reports success only after the operation is committed under $\mathsf{cfg}$.
    \item For \textsf{read}, the client collects responses containing $(v_i,c_i,md)$ for key $k$. Each response is bound to a committed version through $\mathsf{Cert}_{\mathsf{cfg}}(\cdot)$ that the client can validate. The client discards invalid shares using the commitments $\{c_i\}$ and reconstructs $v$ from any $t+1$ valid shares, which is the threshold needed to interpolate a degree-$t$ sharing. If the key is absent, the protocol returns \textsf{NOTFOUND}.
    \item For \textsf{execute}, members run an actively secure MPC computation over their shares and any public inputs, and return a certified output $\mathsf{Cert}_{\mathsf{cfg}}(\textit{reqID},\textit{out},\ell)$, where $\ell$ is the committed log position of the operation.
\end{packeditemize}

\subsubsection{Data Management Protocol}
\label{sec:dpss}
To defend against an adaptive mobile adversary, secrets in \sysname must remain protected even when the adversary moves between enclaves over time. The Protocol layer therefore couples COBRA DPSS~\cite{vassantlalCOBRADynamicProactive2022} maintenance to the DAO-driven configuration lifecycle.

When the DAO commits a membership-changing proposal, such as removal, admission, or readmission, it moves the system into $\mathsf{Reconfig}$ and installs a new configuration $\mathsf{cfg}'$. The Protocol layer starts DPSS maintenance during this reconfiguration window because the set of share holders is changing and share movement must follow the new membership and heterogeneity constraints. Concretely, the protocol runs distributed polynomial generation (DPG) to sample a fresh degree-$t$ zero-sharing polynomial for refresh or a resharing polynomial for the new configuration.

\begin{packeditemize}
\item \emph{Secret resharing.}
To move existing secrets to the new membership set, the committee runs COBRA-style dynamic resharing. Each remaining member contributes a verifiable re-randomization term, producing shares that remain consistent with the same underlying secret but are now held by members of $\mathsf{cfg}'$. Confidentiality is preserved because the adversary controls at most $t$ enclaves at any time, and correctness is enforced through commitments and consistency checks.

\item \emph{Proactive refresh.}
Even without a membership change, the committee periodically or policy-triggeredly refreshes shares by adding a freshly generated degree-$t$ zero-sharing polynomial to the current sharing. This invalidates shares that may have been exposed earlier. The DPG used for refresh is executed inside enclaves and tied to the DAO log so that all correct members refresh the same version.

\item \emph{Secret recovery.}
When a removed enclave is patched and readmitted, it no longer has a valid current share. The protocol runs share recovery so that the recovering enclave obtains helper contributions from other members and reconstructs its own new share under $\mathsf{cfg}'$. Commitments provide verifiability, and the secret $v$ is never reconstructed in the clear.
\end{packeditemize}

During the $\mathsf{Reconfig}$ phase, clients may receive return code \textsf{RECONFIG} and retry after $\mathsf{cfg}'$ has been installed. 

\begin{figure}
    \centering
    \includegraphics[width=.45\textwidth]{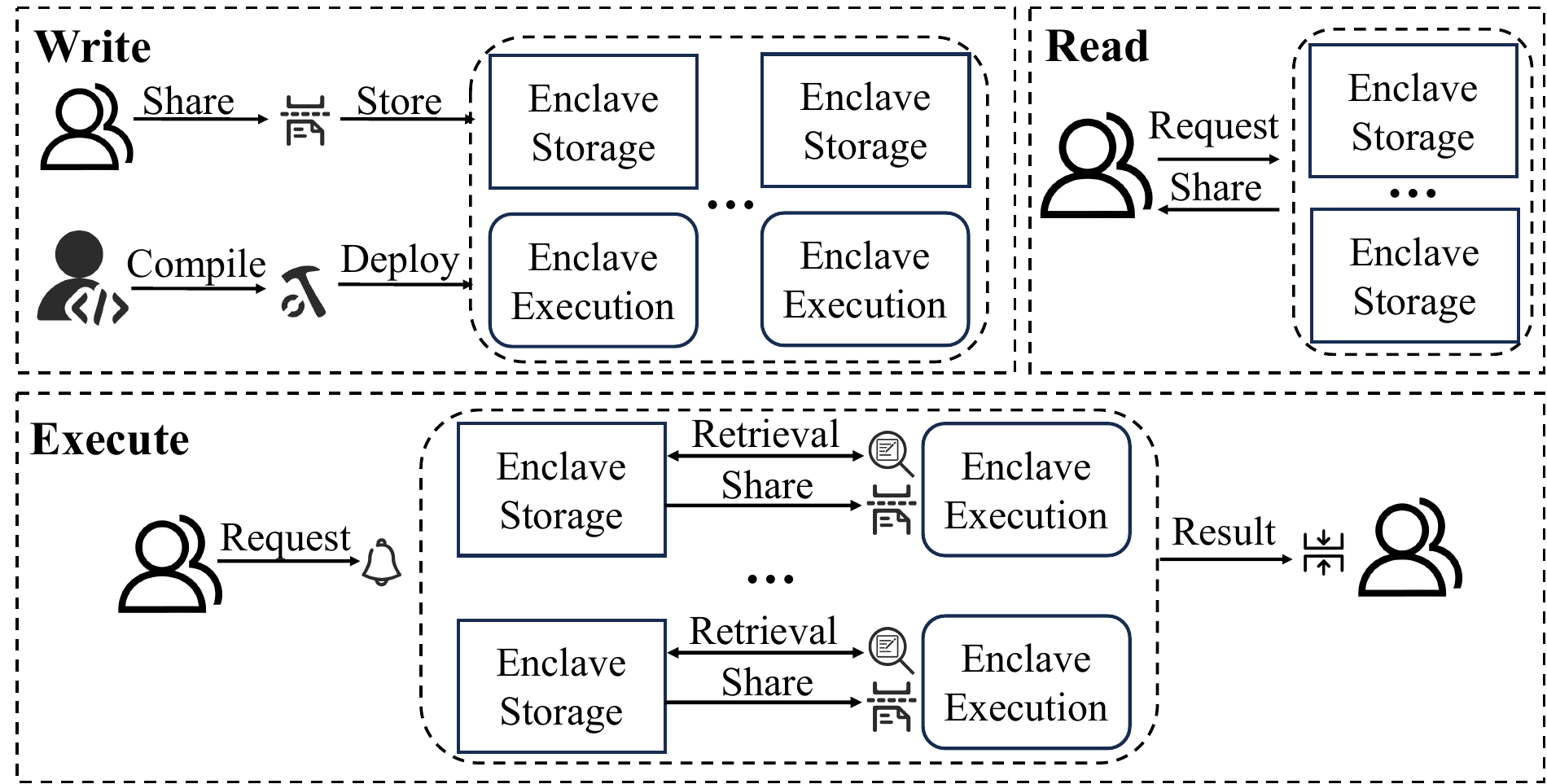}
    \caption{Interfaces of \sysname.}
    \label{fig:TEE-DAO_service}
    \vspace{-6mm}
\end{figure}

\subsubsection{Decentralized Computation Protocol}
\label{sec:mpc}
To support general-purpose computation over secret-shared state, \sysname executes MPC over DPSS-stored shares. An application deploys each computation as a code artifact identified by its digest. The DAO layer governs which artifacts are admissible through $\sigma_{\mathsf{cfg}}$, in particular through the approved code digest $h_{\mathsf{cfg}}$. An $\textsf{execute}(\textit{code},\textit{args})$ request is accepted only when the code digest matches $h_{\mathsf{cfg}}$ or otherwise satisfies the admissibility rule encoded by $\mathsf{isValid}(\cdot,\sigma_{\mathsf{cfg}})$. This prevents unauthorized computations from running inside the committee.

The threat model allows compromised enclaves as well as host and network interference. \sysname therefore cannot rely on semi-honest behavior. It uses actively secure MPC over Shamir sharing, such as protocols proposed by Lindell and Nof~\cite{lindellFrameworkConstructingFast2017a}, with consistency checks and degree-reduction steps that preserve correctness in the presence of up to $t$ Byzantine parties\footnote{\sysname can also be viewed as an enhancement of a pure MPC system, as discussed in Appendix~\ref{appen:mpc}.}. If a computation aborts because of a detected inconsistency, members generate signed evidence bundles $\mathsf{evid}_i$ that can be submitted to the DAO layer as accusations. The DAO layer can then add nodes to the suspect set and trigger reconfiguration according to policy.

\bheading{Certified Outputs and Auditing.}
Upon successful completion, participating members produce a certified output for the client. Members sign, or provide signature shares for, the tuple $(\textit{reqID},\textit{out},\mathsf{cfg},\ell)$, where $\ell$ is the log position of the committed \textsf{execute} operation. The client accepts $\textit{out}$ only if it verifies the corresponding certificate $\mathsf{Cert}_{\mathsf{cfg}}(\textit{reqID},\textit{out},\ell)$. The certificate also supports post-hoc auditing: if an enclave is later marked vulnerable through a status update, clients and auditors can identify outputs certified during the affected configuration window and inspect which transactions depended on the affected committee.

\section{Security Analysis} \label{sec:security}

\subsection{Safety Analysis}
\label{sec:security-safety}

\begin{theorem}[Safety]
\label{thm:safety}
Under (A1--A2), \sysname satisfies Agreement (Def.~\ref{def:agreement}) and
Linearizability (Def.~\ref{def:linearizability}).
\end{theorem}

\begin{proof}[Proof] We prove the following properties by combining the assumed security guarantees of the underlying building blocks (\secref{sec:primitives}) with \sysname’s design. 

\begin{packeditemize}
\item \emph{Agreement.} Given a configuration $\mathsf{cfg}$. By (A1), 
correct committee members commit a unique totally ordered log $\mathsf{Log}_{\mathsf{cfg}}$. This is exactly Def.~\ref{def:agreement}.
    
\item \emph{Linearizability.} Consider an operation $op$ that returns \textsf{SUCCESS}.  By the protocol interface, a successful completion is justified by a commit certificate $\mathsf{Cert}_{\mathsf{cfg}}(op)$ for some configuration $\mathsf{cfg}$. By (A2), such a certificate implies that $op$ is committed under $\mathsf{cfg}$, hence appears at a unique position in $\mathsf{Log}_{\mathsf{cfg}}$. Define the linearization point of $op$ to be this commit position in the global sequence consistent with the committed logs $\mathsf{Log}_{\mathsf{cfg}_0},\mathsf{Log}_{\mathsf{cfg}_1},\ldots$. This yields a single sequential history consistent with the committed logs and preserving non-overlap completion order, establishing Def.~\ref{def:linearizability}.
\end{packeditemize}
\end{proof}

\bheading{Configuration soundness.}
The above safety guarantees hold when configuration changes cannot bypass governance rules: A correct node applies a reconfiguration only if it is both committed and admissible under the current
governance state.

\begin{lemma}[Only committed and admissible proposals take effect]
\label{lem:admissible}
If a correct node applies a governance proposal $p$ while in configuration $\mathsf{cfg}$, then
(i) $\mathsf{Cert}_{\mathsf{cfg}}(p)$ exists, and (ii) $\mathsf{isValid}(p,\sigma_{\mathsf{cfg}})$ holds at the time of application.
\end{lemma}

\begin{proof}[Proof]
Correct nodes apply state transitions only for committed log entries; thus $p$ must be in
$\mathsf{Log}_{\mathsf{cfg}}$ and justified by $\mathsf{Cert}_{\mathsf{cfg}}(p)$ (A1--A2). The governance transition checks
$\mathsf{isValid}(p,\sigma_{\mathsf{cfg}})$ before applying $p$ by definition of admissibility, so (ii) holds.
\end{proof}

\begin{corollary}[Installed configurations satisfy heterogeneity]
\label{cor:hetero}
For every installed configuration $\mathsf{cfg}$, $\mathsf{Hetero}(\mathsf{cfg})$ holds.
\end{corollary}

\begin{proof}[Proof]
Any membership/configuration update is applied only if admissible (Lemma~\ref{lem:admissible}).
Since $\mathsf{Hetero}(\cdot)$ is enforced inside $\mathsf{isValid}(\cdot,\sigma_{\mathsf{cfg}})$, the resulting installed
configuration satisfies $\mathsf{Hetero}(\mathsf{cfg})$.
\end{proof}

\subsection{Liveness Analysis}
\begin{theorem}[Liveness]
\label{thm:liveness}
Under (A1, A3, A4, A5), \sysname guarantees that valid client requests are eventually committed and executed.
\end{theorem}

\begin{proof}[Proof]
Liveness requires physical node availability, consensus progress, and data reconstructability.
First, by (A5), the heterogeneity constraint prevents common-mode failures from compromising $>t$ nodes, guaranteeing that a quorum of $2t+1$ correct nodes always exists.
Second, by (A1), HotStuff overrides leader omission faults to ensure transaction ordering at the protocol layer; by (A3 and A4), the threshold sharing scheme ensures that the valid committee size of $2t+1$ is strictly larger than the reconstruction threshold $t+1$, guaranteeing DPSS reconstruction and MPC execution.
Finally, the DAO layer executes reconfiguration to replace persistently unresponsive nodes, preserving the resilience conditions for the layers above, and satisfying Def.~\ref{def:liveness}.
\end{proof}

\subsection{Proactive Secrecy Analysis}
\label{sec:security-secrecy}

\begin{theorem}[Proactive secrecy]
\label{thm:proactive}
Under (A2--A4), \sysname satisfies Proactive Secrecy (Def.~\ref{def:secrecy}).
\end{theorem}

\begin{proof}[Proof]
Fix two executions that induce the same committed logs
$\mathsf{Log}_{\mathsf{cfg}_0},\mathsf{Log}_{\mathsf{cfg}_1},\ldots$ and the same API outputs, but differ in the underlying
plaintext secrets and private inputs. We show that the adversary's views are computationally indistinguishable.

For each MPC invocation (the $\mathsf{execute}$ interface), by (A4) there exists a simulator producing a
computationally indistinguishable view given only the public committed information and the output.
For secret-shared values used for storage and computation, within any configuration at most $t$ shares
can be exposed at any time, so they reveal no information about the plaintext; across configurations,
whenever a reconfiguration refreshes shares, previously exposed shares do not help recover secrets in later
configurations by (A3). Therefore, conditioned on the same committed logs and the same API outputs,
the adversary's entire view is indistinguishable between the two executions, establishing Def.~\ref{def:secrecy}.
\end{proof}

\section{Performance Evaluation} \label{Evaluation}
We evaluate \sysname from three primary service abstractions: KVS, system membership management, and MPC computation. To provide a comprehensive assessment of system performance, we select the state-of-the-art BFT SMR system, HotStuff~\cite{yinHotStuffBFTConsensus2019}, as the baseline and conduct a comparative analysis with COBRA~\cite{vassantlalCOBRADynamicProactive2022}, a cutting-edge BFT SMR system with confidentiality protection. Our evaluation aims to answer the following questions:
\begin{packeditemize}
    \item \textbf{(Q1) KVS Performance:} How does \sysname perform as a key-value storage service versus traditional BFT SMR systems? (\ssecref{kvstorage_performance})
    \item \textbf{(Q2) Membership Management Overhead:} What is the performance overhead of membership management, including system reconfiguration and member recovery? (\ssecref{management_cost})
    \item \textbf{(Q3) Computation Task Efficiency:} How efficiently does \sysname perform computation tasks, including AES encryption, simple graph search algorithms, and machine learning, under varying numbers of nodes? (\ssecref{computation_efficiency})
    \item \textbf{(Q4) Heterogeneous TEE Deployment:}
    How does \sysname perform under mixed deployments spanning different TEE implementations?
(\ssecref{heterogeneous_performance})
\end{packeditemize}

\begin{figure*}[ht]
    \centering

    \begin{subfigure}[b]{0.31\textwidth}
        \centering
        \includegraphics[width=\textwidth]{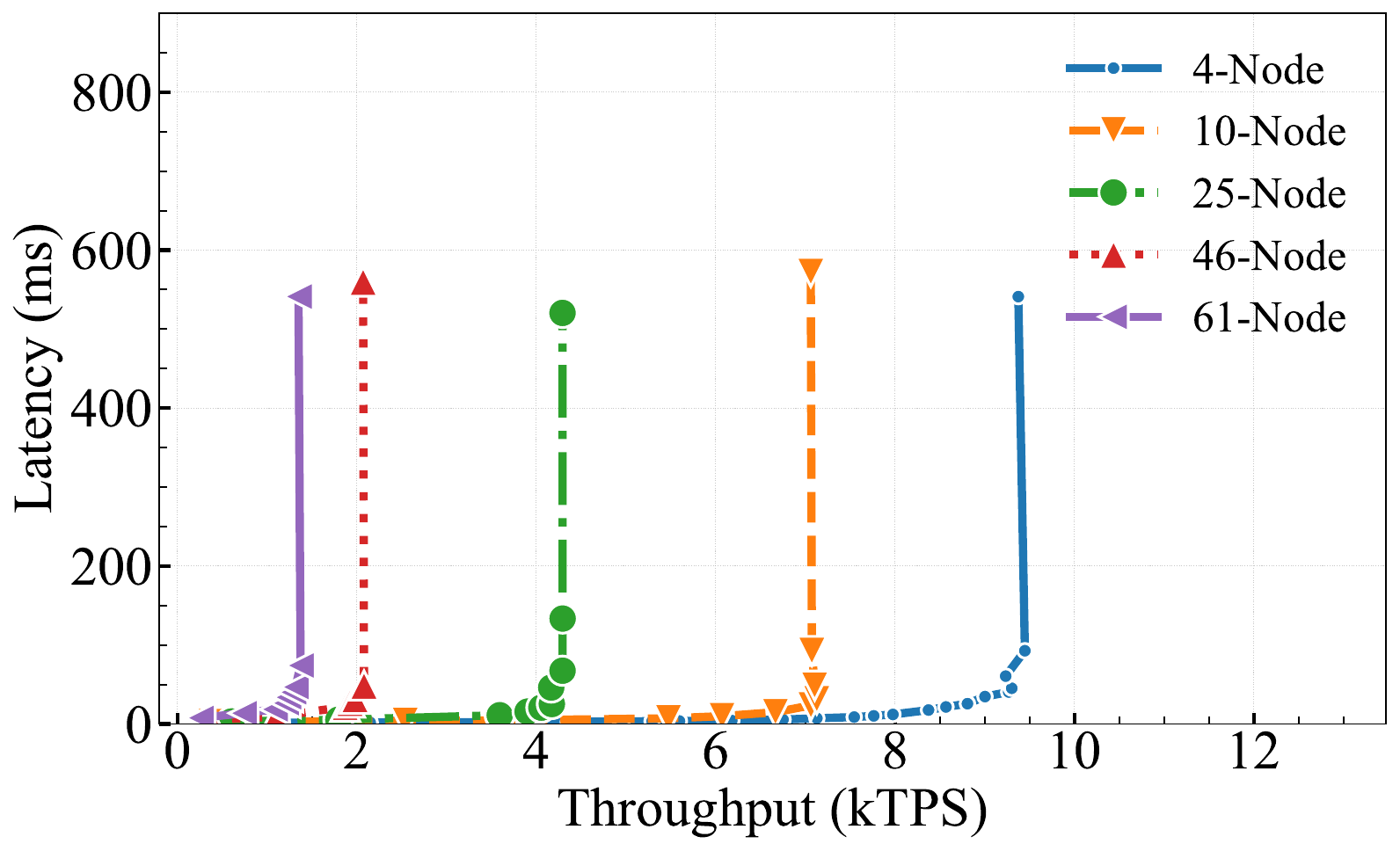}
        \caption{COBRA (Write)}
    \end{subfigure}
    \hfill
    \begin{subfigure}[b]{0.31\textwidth}
        \centering
        \includegraphics[width=\textwidth]{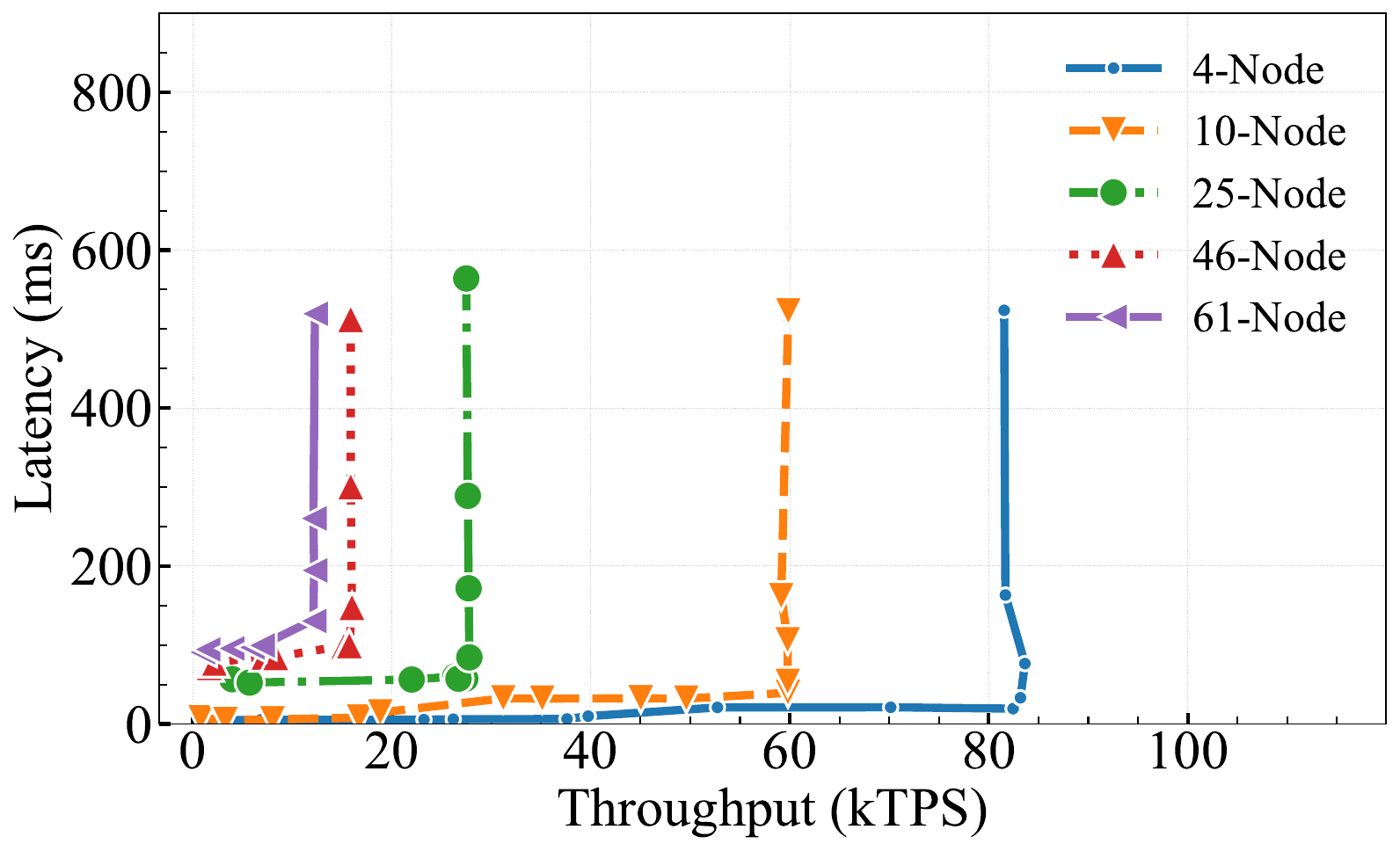}
        \caption{HotStuff (Write)}
    \end{subfigure}
    \hfill
    \begin{subfigure}[b]{0.31\textwidth}
        \centering
        \includegraphics[width=\textwidth]{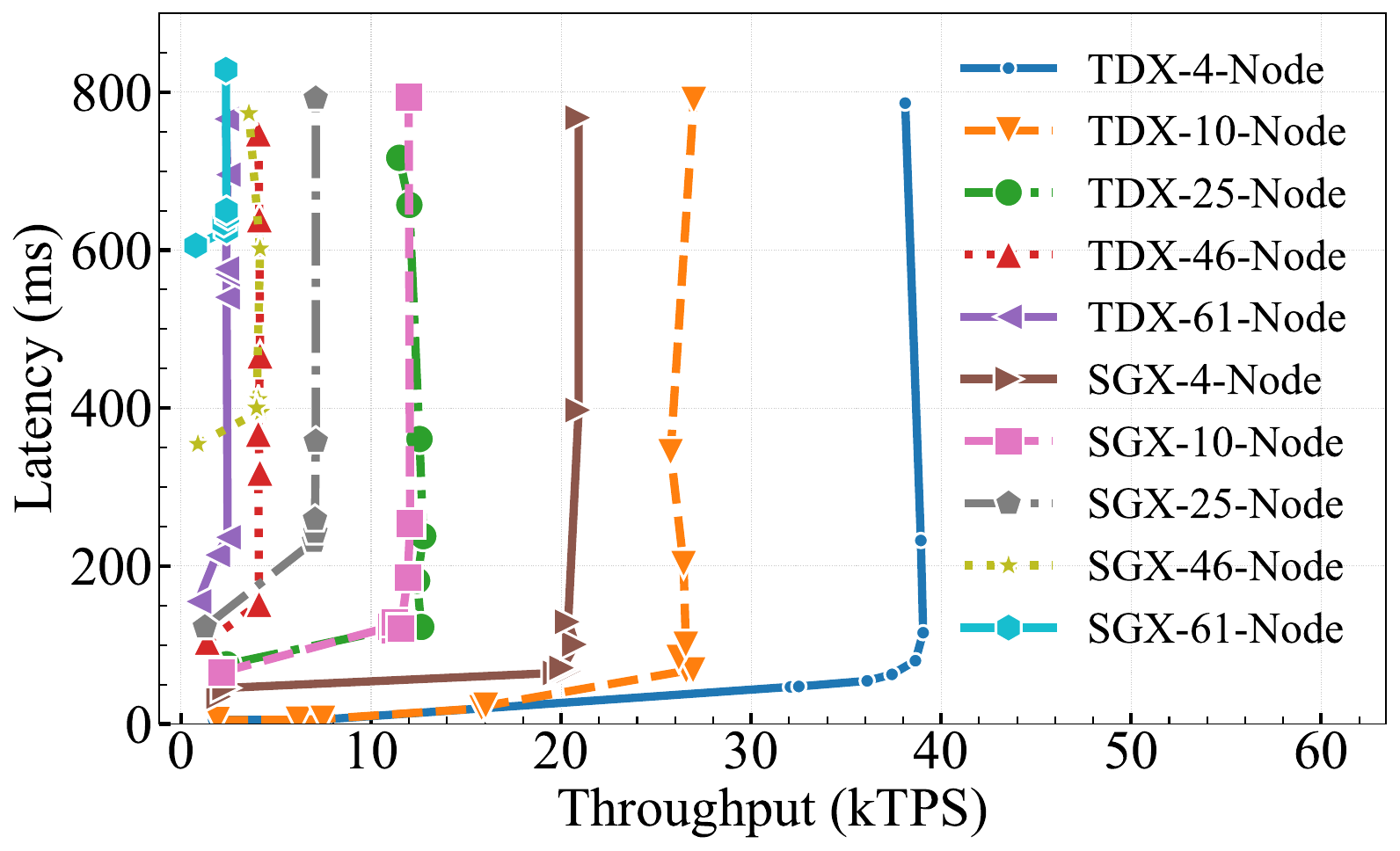}
        \caption{Ours (Write)}
    \end{subfigure}

    \vspace{1mm}

    \begin{subfigure}[b]{0.31\textwidth}
        \centering
        \includegraphics[width=\textwidth]{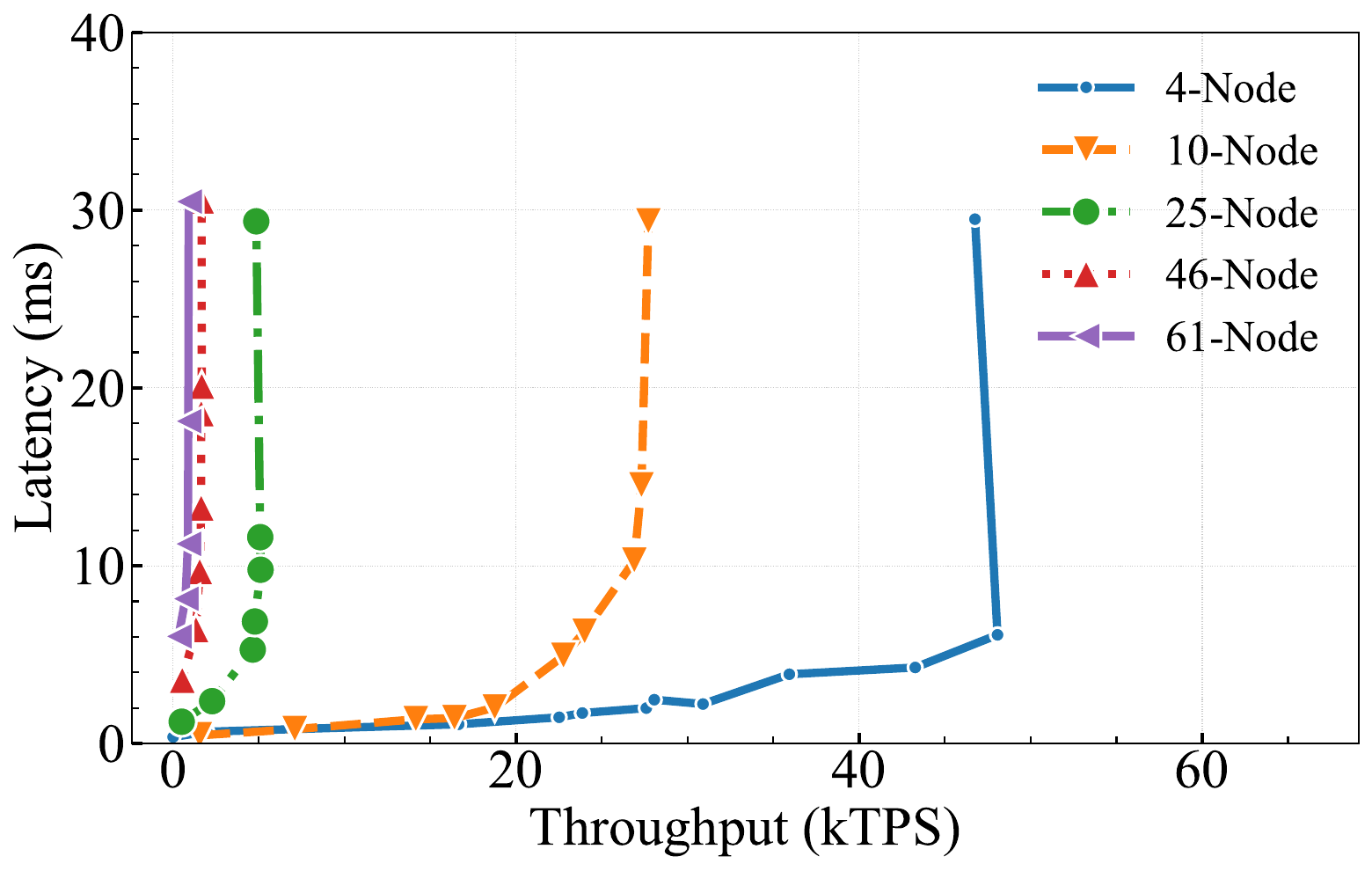}
        \caption{COBRA (Read)}
    \end{subfigure}
    \hfill
    \begin{subfigure}[b]{0.31\textwidth}
        \centering
        \includegraphics[width=\textwidth]{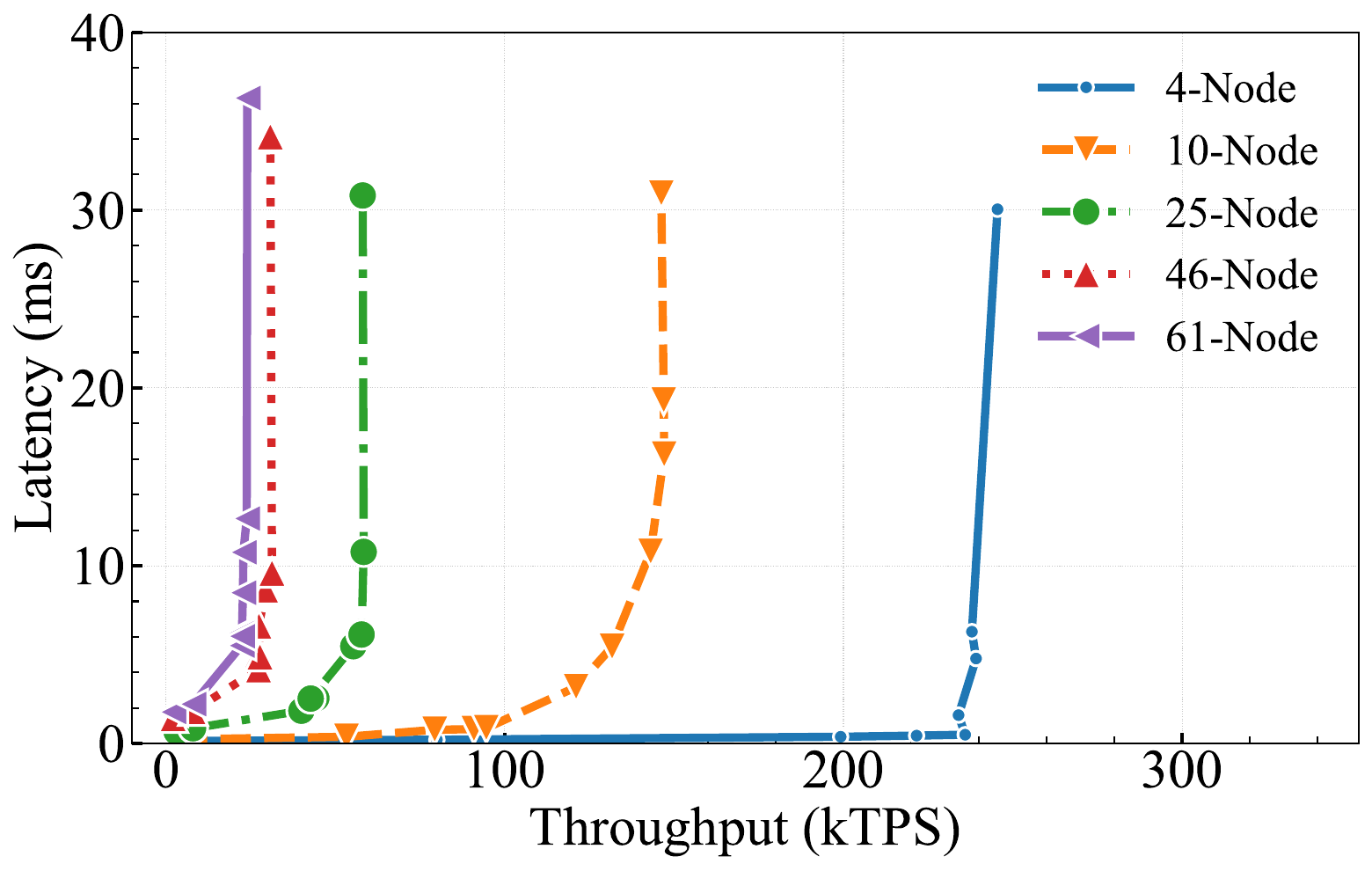}
        \caption{HotStuff (Read)}
    \end{subfigure}
    \hfill
    \begin{subfigure}[b]{0.31\textwidth}
        \centering
        \includegraphics[width=\textwidth]{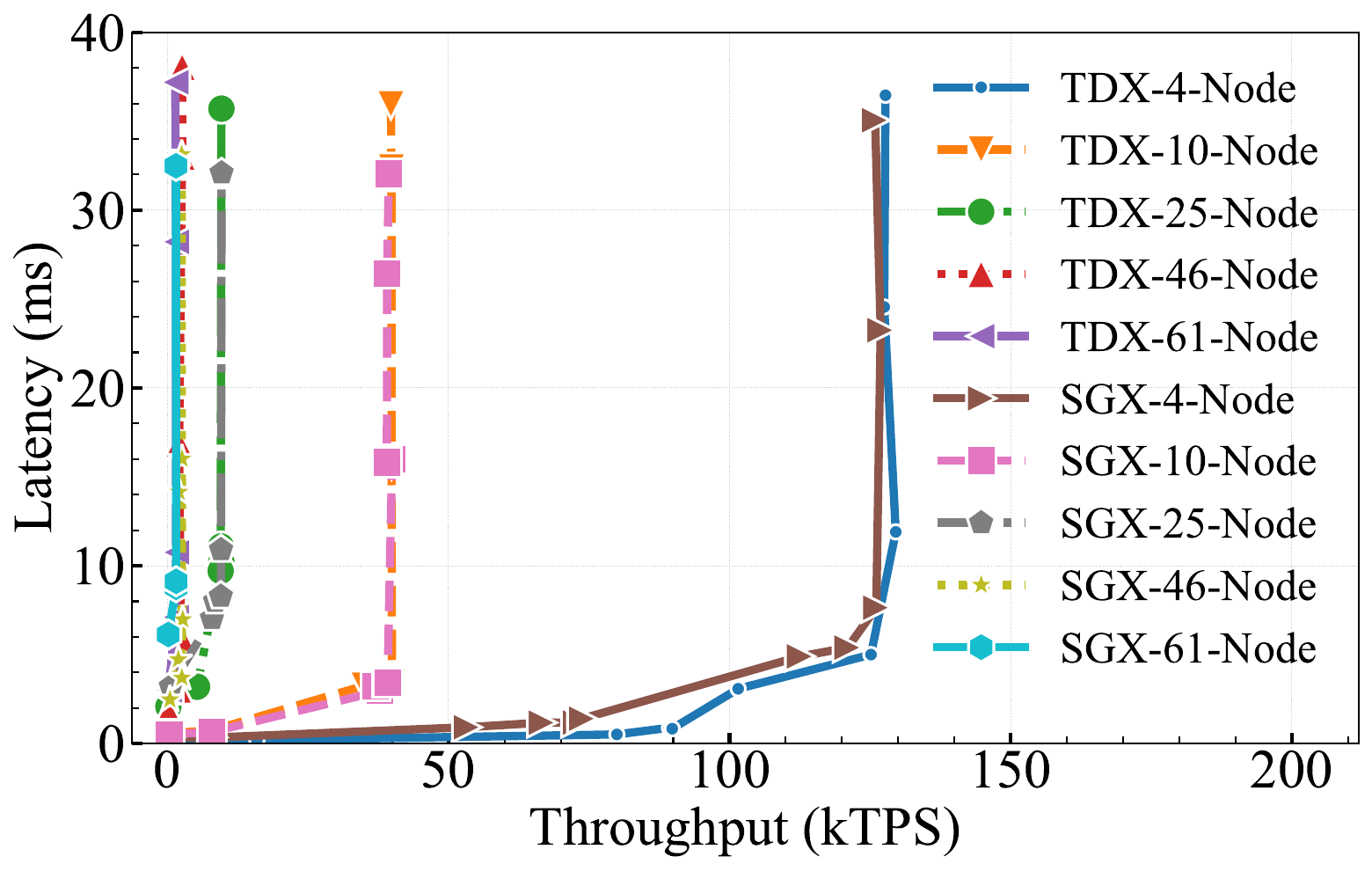}
        \caption{Ours (Read)}
    \end{subfigure}

    \caption{Write and read throughput/latency.}
    \label{fig:tee-dao_throughput}
    \vspace{-4mm}
\end{figure*}

\subsection{Implementation and Experiment Setup}
\label{experiment_setup}

\bheading{Implementation.} 
We implement a prototype of \sysname atop several commodity TEE implementations, i.e., Intel SGX, Intel TDX, and Hygon CSV, with the latest DPSS techniques, COBRA~\cite{vassantlalCOBRADynamicProactive2022}, and the versatile MPC compiler MP-SPDZ~\cite{kellerMPSPDZVersatileFramework2020}. We implement COBRA with linear secret sharing from scratch atop HotStuff~\cite{HotStuffLibhotstuff2025}. We also add functionality to adapt the secret sharing storage format to the Montgomery representation of MP-SPDZ for computation. The whole system comprises a total of 10,353 lines of C++ code.

\bheading{Experiment Setup.} We evaluate the performance of \sysname in two VM-based TEEs (i.e., Intel TDX and Hygon CSV) and process-based TEEs (i.e., Intel SGX). \sysname uses Intel Xeon Platinum 8369B CPU (2.7GHz) as Intel SGX platforms, Intel Xeon Emerald Rapids CPU (2.7GHz) for Intel TDX platforms, and Hygon C86-3G 7390 CPU (3.0 GHz) for Hygon CSV. For comparison, COBRA and HotStuff are deployed on a standard server equipped with an Intel Xeon Platinum 8575C CPU (3.2GHz). All platforms are configured with 8 vCPUs and 32GB of RAM, running a Linux kernel version above v5.10. The Intel SGX platform is configured with SGX v2.0 and 16GB of physical memory  (maximum EPC 16GB). The experiments were conducted in a controlled environment with a round-trip time (RTT) of 0.05–0.07ms and a bandwidth of 5 Gbps, representing a common deployment scenario for \sysname in cloud data centers.

All protocols are evaluated under identical settings. Each request carries a 1024-byte payload, and the consensus batch size is set to 400. For secret sharing, both COBRA and \sysname use Feldman’s linear commitment scheme~\cite{feldmanPracticalSchemeNoninteractive1987}. We conduct three sets of experiments to address the questions outlined above. We focus on Intel TDX and SGX for direct comparison, and defer the evaluation of CSV to Appendix~\ref{appen:csv_expriment} due to space constraints.

\subsection{KV Storage Performance}
\label{kvstorage_performance}
To address \textbf{Q1}, we evaluate the performance of \sysname as a KVS in terms of write and read requests.

\bheading{Write requests.} \figref{fig:tee-dao_throughput}.a-c presents the latency and throughput of three systems given write requests of 1KB secret data with 4 to 61 nodes. The results indicate that HotStuff achieves the highest throughput as it does not perform secret-sharing. Specifically, as the number of nodes increases from 4 to 61, HotStuff's throughput declines from 81 kTPS to 12 kTPS. In comparison, \sysnameTDX outperforms \sysnameSGX, benefiting from the larger computing memory provided by TDX VMs and its support for para-virtualization and SR-IOV \cite{IntelTDXConnect}. At $n=4$, \sysnameTDX achieves 54\% of HotStuff's throughput, while \sysnameSGX achieves 47\%. COBRA exhibits the lowest throughput, primarily due to the higher complexity of its consensus protocol and the overhead introduced by frequent locking in its implementation. For instance, at $n=4$, COBRA achieves only 9 kTPS, which is 11\% of HotStuff's throughput. 

With an increasing number of nodes, the overhead introduced by TEEs causes \sysname's throughput to decline more rapidly than that of the other systems. At $n=61$, \sysnameTDX achieves 20\% of HotStuff's throughput, while \sysnameSGX achieves 18\%. Nevertheless, both \sysnameTDX and \sysnameSGX significantly outperform COBRA. Specifically, \sysnameTDX achieves a throughput 1.8 times higher than COBRA, while \sysnameSGX achieves 1.6 times higher.

\begin{table}[t]
    \centering
    \setlength{\abovecaptionskip}{2pt}
    \renewcommand{\arraystretch}{0.95}
    \setlength{\tabcolsep}{4pt}
    \caption{Reshare latency (s) for storage with 1K entries.}
    \label{tab:reshare_breakdown_horizontal}

    \begin{tabular*}{\columnwidth}{@{\extracolsep{\fill}}lccccc@{}}
        \toprule
        System & Metric & $n=4$ & $n=16$ & $n=25$ & $n=31$ \\
        \midrule
        \multirow{2}{*}{TDX}
            & Poly. Gen.   & 1.899 & 10.296 & 21.845 & 32.392 \\
            & State Recon. & 0.466 & 1.345  & 2.546  & 3.763 \\
        \midrule
        \multirow{2}{*}{SGX}
            & Poly. Gen.   & 3.652 & 17.209 & 33.685 & 45.469 \\
            & State Recon. & 0.503 & 1.412  & 2.807  & 4.243 \\
        \midrule
        \multirow{2}{*}{COBRA}
            & Poly. Gen.   & 3.970 & 11.591 & 21.343 & 30.382 \\
            & State Recon. & 0.963 & 1.150  & 1.573  & 1.680 \\
        \bottomrule
    \end{tabular*}
    \vspace{-4mm}
\end{table}

\bheading{Read requests.} \figref{fig:tee-dao_throughput}.d-f illustrates the latency and throughput of read requests, which follow a similar trend to write requests. The throughput of \sysnameTDX and \sysnameSGX ranks second only to HotStuff and consistently outperforms COBRA. Unlike write operations, read operations do not involve consensus protocols, so I/O overhead and secret-share combination become the dominator. 

The experimental results show that as the number of nodes increases from 4 to 61, HotStuff's throughput decreases from 240 kTPS to 23 kTPS. Due to the lower communication and computational complexity of read operations, \sysnameSGX and \sysnameTDX achieve nearly identical performance. When $n=4$, both \sysnameTDX and \sysnameSGX achieve 45\% of HotStuff's throughput, while COBRA achieves only 19\%. However, as the number of nodes increases, the I/O overhead introduced by TEEs significantly impacts performance. When $n=61$, the throughput of \sysnameTDX and \sysnameSGX decreases to 1.5 kTPS, slightly higher than COBRA's throughput of 0.9 kTPS.

\begin{figure}[t]
    \centering
    \begin{subfigure}[b]{0.22\textwidth}
        \centering
        \includegraphics[width=\textwidth]{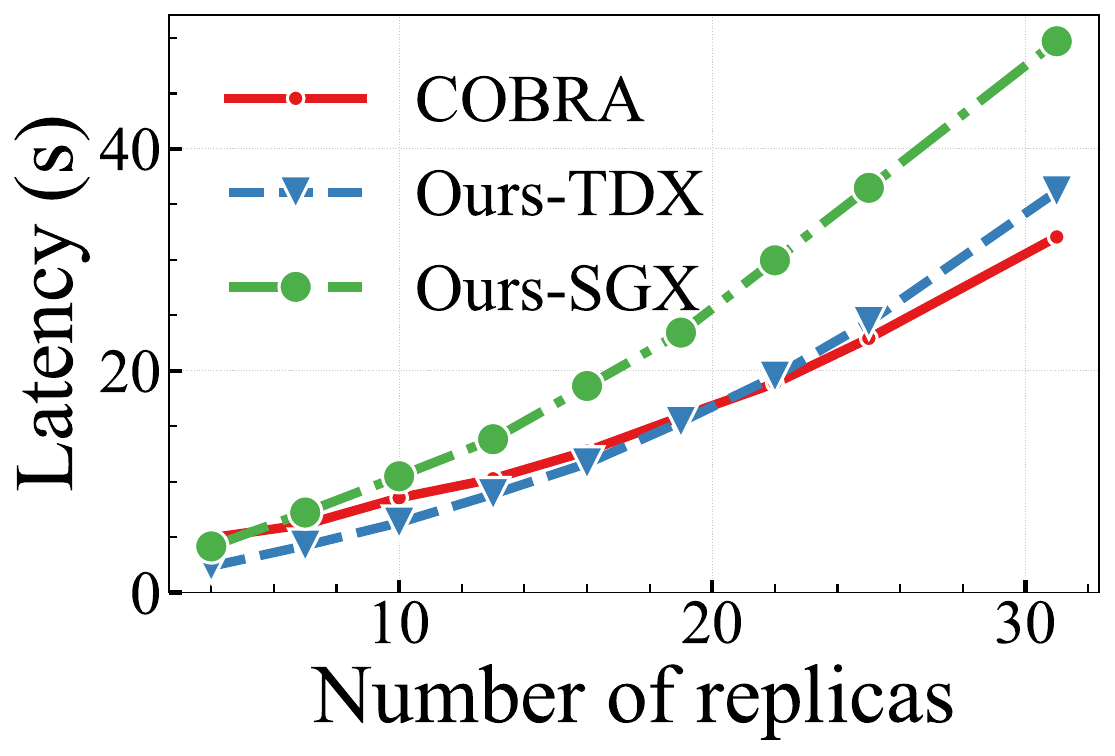}
        \caption{Reconfiguration Cost}
        \label{fig:tee-dao_reconfig_cost}
    \end{subfigure}
    \begin{subfigure}[b]{0.22\textwidth}
        \centering
        \includegraphics[width=\textwidth]{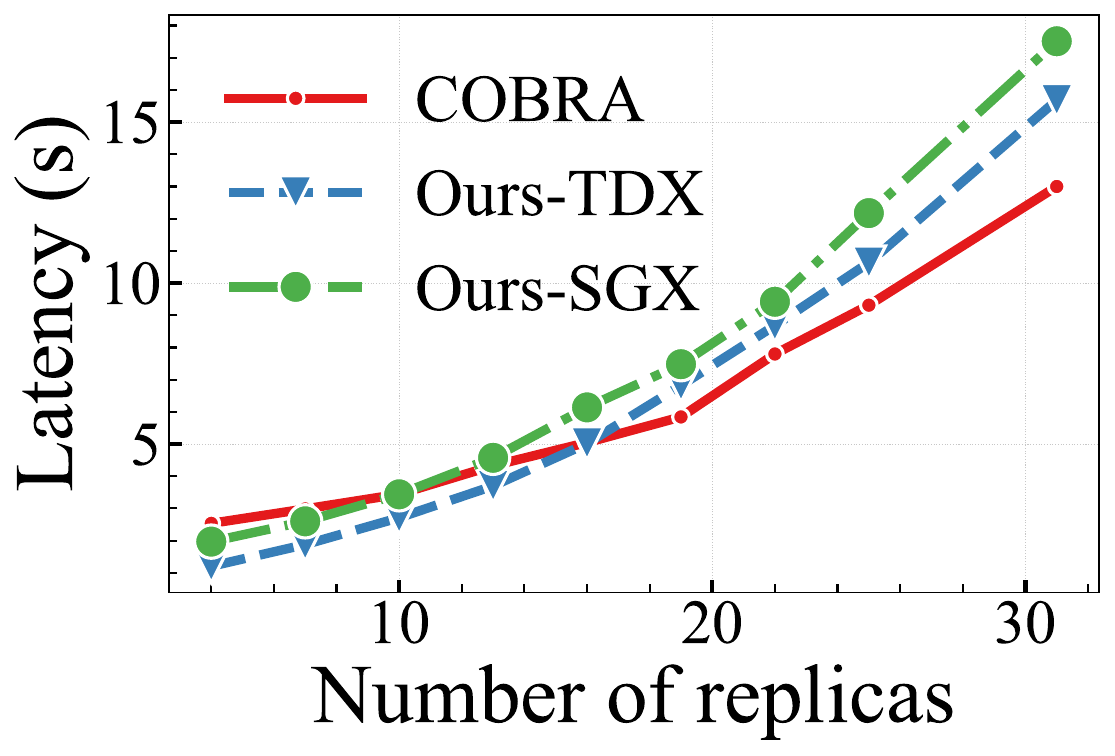}
        \caption{Recovery Cost}
        \label{fig:tee-dao_recovery_cost}
    \end{subfigure}
    \caption{Reconfiguration and recovery latency.}
    \label{fig:tee-dao_costs}
    \vspace{-4mm}
\end{figure}

\subsection{Management Cost}
\label{management_cost}
To address \textbf{Q2}, we evaluate the system management cost, including two key operations: polynomial generation during consensus for reconfiguration and secret state reconstruction. 
We measure the cost of adding (recovery) and removing (reconfiguration) a server from a system containing 1K secret data entries, with configurations ranging from 4 to 31 nodes. Since HotStuff does not provide confidentiality protection, our evaluation focuses on comparing \sysname and COBRA.

\bheading{Reconfiguration latency.} \figref{fig:tee-dao_reconfig_cost} presents the reconfiguration latency. When there are fewer than 22 nodes, \sysnameTDX outperforms COBRA in shorter latency, while \sysnameSGX achieves lower overhead than COBRA only when $n = 4$. This performance advantage is attributed to \sysname's foundation on HotStuff, which has lower protocol complexity compared to the BFT-SMaRt system used by COBRA~\cite{sousaByzantineConsensusBFT2012}. Specifically, when $n = 4$, \sysnameTDX achieves 47\% of COBRA's reconfiguration latency, and \sysnameSGX achieves 84\%. 

When the number of nodes increases, the I/O overhead introduced by TEEs significantly impacts \sysname's performance. \tabref{tab:reshare_breakdown_horizontal} shows that the time required for polynomial generation increases significantly with the number of nodes, surpassing that of COBRA. Moreover, the growth rate of \sysnameSGX is faster than that of \sysnameTDX. When $n = 22$, the reconfiguration latency of \sysnameTDX begins to exceed that of COBRA.

\begin{table}[t]
    \centering
    \setlength{\abovecaptionskip}{2pt}
    \renewcommand{\arraystretch}{0.95}
    \setlength{\tabcolsep}{4pt}
    \caption{Recovery latency (s) for storage with 1K entries.}
    \label{tab:recovery_breakdown_horizontal}

    \begin{tabular*}{\columnwidth}{@{\extracolsep{\fill}}lccccc@{}}
        \toprule
        System & Metric & $n=4$ & $n=16$ & $n=25$ & $n=31$ \\
        \midrule
        \multirow{2}{*}{TDX}
            & Poly. Gen.   & 1.026 & 4.606 & 9.788 & 14.414 \\
            & State Recon. & 0.179 & 0.411 & 0.845 & 1.265 \\
        \midrule
        \multirow{2}{*}{SGX}
            & Poly. Gen.   & 1.665 & 5.128 & 10.226 & 14.620 \\
            & State Recon. & 0.305 & 1.011 & 1.947 & 2.890 \\
        \midrule
        \multirow{2}{*}{COBRA}
            & Poly. Gen.   & 2.267 & 4.311 & 8.004 & 11.347 \\
            & State Recon. & 0.268 & 0.736 & 1.308 & 1.652 \\
        \bottomrule
    \end{tabular*}

    \vspace{-2mm}
\end{table}

\bheading{Recovery latency.} \figref{fig:tee-dao_recovery_cost} illustrates the recovery latency for a faulty node in clusters containing 4 to 31 nodes. The results show that when the number of nodes is fewer than 16, \sysnameTDX achieves a shorter recovery time compared to COBRA. Additionally, \sysnameSGX slightly outperforms COBRA when the number of nodes is fewer than 13, while still exhibiting higher recovery latency than \sysnameTDX. However, as the number of nodes increases, COBRA gradually outperforms both \sysnameTDX and \sysnameSGX. 

\tabref{tab:recovery_breakdown_horizontal} shows that the polynomial generation time for \sysname increases rapidly with the increasing number of nodes, with \sysnameSGX exhibiting a faster growth rate than \sysnameTDX. At $n = 31$, the recovery latency of \sysnameTDX (resp. \sysnameSGX) is 20\% (resp. 34\%) higher than COBRA.

\subsection{Computation Overhead}
\label{computation_efficiency}

\begin{figure}[t]
    \centering
    \includegraphics[width=.47\textwidth]{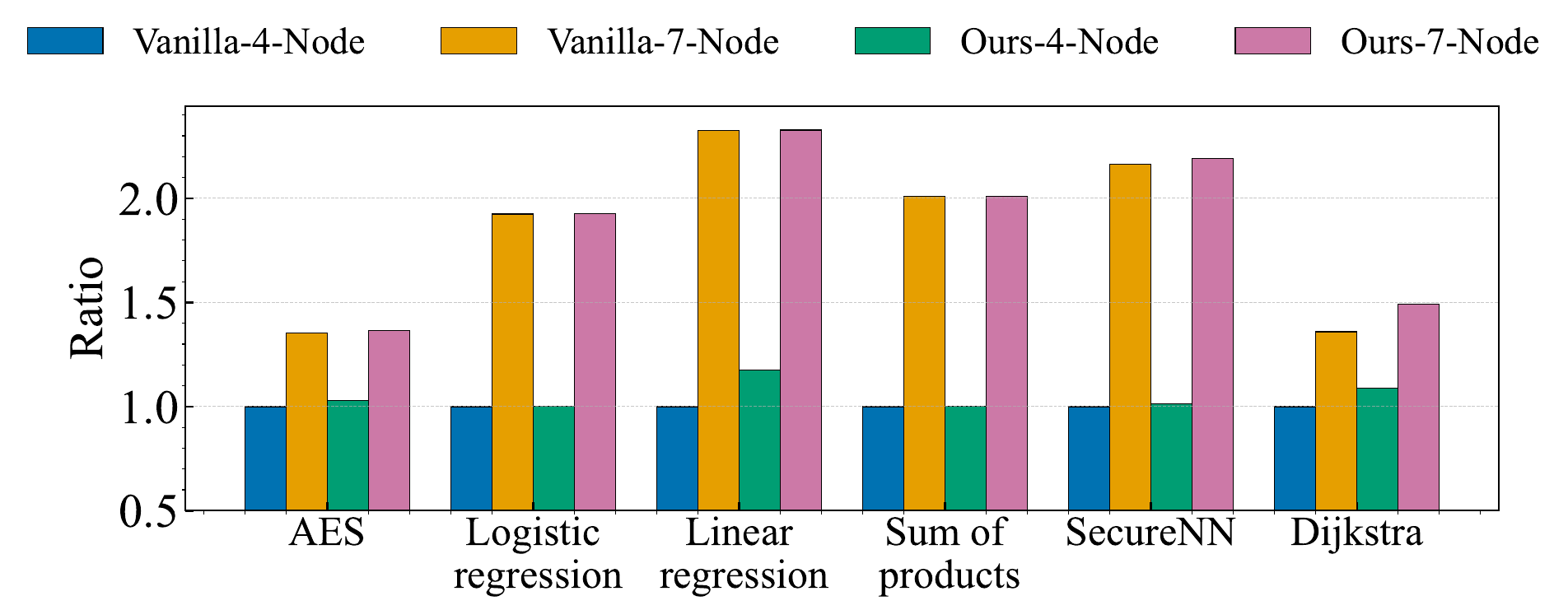}
    \caption{Normalized computational cost.}
    \label{fig:tee-dao_comp_cost}
    \vspace{-4mm}
\end{figure}

To address \textbf{Q3}, we compare \sysnameTDX with vanilla MPC using Shamir's secret sharing protocol under both malicious adversary and honest majority models. \figref{fig:tee-dao_comp_cost} presents a normalized comparison of computation time between \sysname and vanilla MPC, using the computation time of vanilla MPC with $n = 4$ as the baseline. The experiment evaluates both systems under $n = 4$ and $n = 7$, leveraging benchmark tasks provided by the MP-SPDZ library~\cite{kellerMPSPDZVersatileFramework2020}. These tasks include AES encryption, Logistic Regression, Linear Regression, Sum of Products, SecureNN, and the Dijkstra algorithm.

The results demonstrate that \sysname introduces low additional latency for most computation tasks, indicating acceptable performance overhead. Specifically, for tasks such as SecureNN and Sum of Products, \sysname incurs negligible overhead. For other tasks, such as AES encryption and the Dijkstra algorithm, the performance overhead remains below 10\%. Although the overhead for Logistic Regression is slightly higher, it is still controlled within 18\%. These data indicate that \sysname introduces relatively low performance overhead, even for computationally intensive tasks, making it a practical solution for secure collaborative computation.

\subsection{Performance under Heterogeneous TEEs}
\label{heterogeneous_performance}
To address \textbf{Q4}, we further evaluate \sysname under mixed-TEE deployments. The 4-node committee contains one SGX, one CSV, and two TDX nodes, while the 7-node committee contains two SGX, two CSV, and three TDX nodes. These TEEs differ substantially in their isolation models and virtualization architectures, allowing us to evaluate whether \sysname can operate efficiently across heterogeneous trust domains without protocol redesign.

\newlength{\ncolwidth}
\setlength{\ncolwidth}{8pt} 
\begin{table}[t]
\centering
\footnotesize
\caption{Heterogeneous committee performance.}
\label{tab:hetero_eval}
\renewcommand{\arraystretch}{1.0}
\setlength{\tabcolsep}{3pt}
\scalebox{0.97}{
\begin{tabular*}{\columnwidth}{@{\extracolsep{\fill}}c p{\ncolwidth} cccc@{}}
\toprule
System & \centering $n$ &
Write (kTPS/ms) & Read (kTPS/ms) & Recovery (s) & Reconfig. (s)\\
\midrule
SGX    & \multirow{4}{*}{\centering 4}
        & 20.2/71.1 & 115.1/8.0 & 1.970 & 4.155 \\
CSV    &
        & 22.5/64.1 & 120.1/3.6 & 2.133 & 4.078 \\
TDX    &
        & 36.1/54.5 & 122.8/5.3 & 1.205 & 2.365 \\
Hetero &
        & 25.1/70.4 & 121.0/5.8 & 1.731 & 3.251 \\
\midrule
SGX    & \multirow{4}{*}{\centering 7}
        & 14.4/119.6 & 57.4/2.6 & 2.600 & 7.208 \\
CSV    &
        & 18.4/85.6 & 61.7/4.3 & 3.262 & 7.513 \\
TDX    &
        & 33.5/64.3 & 62.7/1.9 & 1.879 & 4.221 \\
Hetero &
        & 20.7/79.0 & 61.2/2.0 & 2.586 & 5.541 \\
\bottomrule
\end{tabular*}}
\vspace{-3mm}
\end{table}

Table~\ref{tab:hetero_eval} shows that heterogeneous deployment performance falls between the best homogeneous deployment (\sysnameTDX) and the slower homogeneous deployment (\sysnameSGX). This behavior is expected because consensus, share maintenance, and reconfiguration involve all committee members, causing the slowest TEE type to partially determine the critical path. For example, TeeDAO-Hetero achieves 25.1~kTPS write throughput at $n=4$ and 20.7~kTPS at $n=7$, both remaining between \sysnameSGX and \sysnameTDX. Recovery and reconfiguration exhibit similar behavior. By contrast, read performance is less affected because reads do not require consensus.

\section{Application Cases}
\label{sec:apps}
We illustrate how \sysname supports common privacy-sensitive services by presenting three detailed cases of secret storage~\cite{AWSKeyManagement, msmbaldwinAzureKeyVaulta}, key custody/signing~\cite{SafeWallet, fireblocks}, and cross-organization analytics~\cite{DBLP:conf/asiacrypt/LepointPRST21}. 

\bheading{Secret KV Store.}
\label{sec:app-kv}
Secret KV store provides APIs to store and retrieve secret values, often with versioning and rotation~\cite{AWSKeyManagement, msmbaldwinAzureKeyVaulta}. It can be realized by \sysname by storing each item as proactively refreshed shares across a heterogeneous committee. A client can (i) upload a plaintext secret to an admitted enclave, which immediately shares it, or (ii) client-split into shares and upload only shares. Storage uses \texttt{write(k, share)}; retrieval uses \texttt{read(k)} plus a policy-controlled reconstruction workflow that avoids concentrating a long-lived secret at any single replica, while allowing users to constrain which TEE types may hold shares.

\bheading{Multisig Wallet.}
\label{sec:app-wallet}
For digital wallet and transaction signing, a widely deployed pattern is split-key authorization: control is distributed either via on-chain multisig wallets or via MPC custody services where the private key is never present in full and is represented by distributed shares~\cite{SafeWallet,fireblocks}.
\sysname generalizes the split-key pattern to long-running, reconfigurable committees under heterogeneity constraints. The signing key is stored as shares (\texttt{write}); a signing request triggers an \texttt{execute} workflow where committee members produce signature shares inside admitted TEEs and the leader combines them into a full signature without reconstructing the private key. This supports proactive refresh and reconfiguration, which are not addressed by static multisig or single-domain KMS-style custody.

\bheading{Privacy-Preserving Data Analytics.}
\label{sec:app-analytics}
Cross-organization analytics increasingly use MPC-style designs to compute aggregates without sharing raw records, e.g., Google Private Join and Compute~\cite{DBLP:conf/asiacrypt/LepointPRST21}.
\sysname supports this as a stateful service: each party uploads data as shared state (\texttt{write}); an analytics job runs as an \texttt{execute} task compiled into an MPC protocol by the committee; outputs are released via \texttt{read} under policy. Compared with pure MPC deployments, \sysname provides a uniform interface and long-running committee management (admission by attested identity, heterogeneity constraints, and proactive refresh) as the platform layer.

\section{Related Works} \label{Related}

\bheading{Rollback resilience.} 
In rollback attacks, an adversary reverts the enclave’s persistent state (or the external storage it relies on) to an older, valid state, leading to violation of state freshness/integrity guarantees. Distributed TEE nodes can assist one another in preserving the freshness of their states when not all TEEs are crashed. ROTE~\cite{mateticROTERollbackProtection2017} distributes rollback protection across multiple SGX processors via a broadcast protocol, while Narrator~\cite{niuNARRATORSecurePractical2022, narrator-pro, tiks} and Nimble~\cite{angelNimbleRollbackProtection2023} rely on a TEE-backed or blockchain-backed ledger to detect storage rollback. RR and its TEEMS metadata service~\cite{dinisRRFaultModel2023} formalize the restart–rollback fault model and adapt crash-tolerant replication protocols to TEEs with external state. Achilles~\cite{Achilles} realizes roll-back resilient recovery for consensus protocols. 

\bheading{Availability enhancement.} TEE-replicated services execute consensus or application logic inside enclaves to enhance availability in the presence of faulty hosts: Engraft~\cite{wangENGRAFTEnclaveguardedRaft2022} protects Raft with SGX, and CCF~\cite{DBLP:journals/pvldb/HowardAACCCDFJK23} runs replicated services in TEEs backed by an auditable ledger. These systems largely assume that each TEE provides integrity and confidentiality guarantees, focusing on host crashes or protocol faults rather than complete TEE compromises.

\bheading{Heterogeneous TEE systems for compromised TEEs.} These systems assume a stronger adversary that may fully break a few particular TEE implementations and therefore distribute trust across different TEE types. Dauterman et al.~\cite{emmadautermanReflectionsTrustingDistributed2022} propose splitting cryptographic trust across independent hardware roots so that security holds as long as at least one root remains honest. Connell et al. ~\cite{connellSecretKeyRecovery2024} propose SVR3 for key recovery by storing users’ key material across Nitro, SEV-SNP, and SGX enclaves in different clouds. Existing systems cannot handle \textit{adaptive mobile adversaries}. In contrast, \sysname provides a general framework to manage membership, resharing, and revocation of heterogeneous TEE committees for long-lived application state. 

\bheading{Hybrid Systems of TEE and MPC.}
Several hybrid systems~\cite{riazi2018chameleon, luCorrelatedRandomnessTeleportation2021, wuHybridTrustMultiparty2022, huEfficientPracticalMultiparty2025} combine TEEs and MPC to leverage the strengths of both paradigms. Specifically, prior work~\cite{riazi2018chameleon, luCorrelatedRandomnessTeleportation2021} treats TEE-enabled nodes as semi-honest MPC parties, using TEE integrity guarantees to accelerate MPC. Later, Wu et al.~\cite{wuHybridTrustMultiparty2022} and Hu et al.~\cite{huEfficientPracticalMultiparty2025} further adapt the division of tasks between TEE and MPC according to security and cost constraints. 
However, they do not consider compromised TEEs. In contrast, \sysname targets long-lived MPC services under a stronger threat model of \textit{adaptive mobile adversaries}, tolerating full compromise of any TEE type and its host.

\section{Conclusion} \label{Conclusion}
We propose \sysname, a three-layer framework that manages a committee of heterogeneous TEEs as a distributed-trust infrastructure for long-running confidential services. It also 
provides a minimal,  uniform API surface to support various confidential computing services such as KVS, cryptocurrency wallets, and multi-party collaborative analytics services. Our experimental evaluation results reveal that \sysname achieved high efficiency in KV Store and management while imposing a reasonable overhead on system management and MPC computation execution.

\clearpage

\bibliographystyle{IEEEtran}
\bibliography{ref}

\begin{thebibliography}{100}
\providecommand{\url}[1]{#1}
\csname url@samestyle\endcsname
\providecommand{\newblock}{\relax}
\providecommand{\bibinfo}[2]{#2}
\providecommand{\BIBentrySTDinterwordspacing}{\spaceskip=0pt\relax}
\providecommand{\BIBentryALTinterwordstretchfactor}{4}
\providecommand{\BIBentryALTinterwordspacing}{\spaceskip=\fontdimen2\font plus
\BIBentryALTinterwordstretchfactor\fontdimen3\font minus \fontdimen4\font\relax}
\providecommand{\BIBforeignlanguage}[2]{{%
\expandafter\ifx\csname l@#1\endcsname\relax
\typeout{** WARNING: IEEEtran.bst: No hyphenation pattern has been}%
\typeout{** loaded for the language `#1'. Using the pattern for}%
\typeout{** the default language instead.}%
\else
\language=\csname l@#1\endcsname
\fi
#2}}
\providecommand{\BIBdecl}{\relax}
\BIBdecl

\bibitem{IntelSoftwareGuard}
``Intel{\textregistered} software guard extensions (intel{\textregistered} sgx),'' https://www.intel.com/content/www/us/en/architecture-and-technology/software-guard-extensions.html, 2015.

\bibitem{IntelTrustDomain}
``Intel{\textregistered} trust domain extensions (intel{\textregistered} tdx),'' https://www.intel.com/content/www/us/en/developer/tools/trust-domain-extensions/overview.html, 2023.

\bibitem{ltdArmConfidentialCompute}
``Arm confidential compute architecture,'' https://www.arm.com/architecture/security-features/arm-confidential-compute-architecture, 2021.

\bibitem{AMDSecureEncrypted}
``Amd secure encrypted virtualization (sev),'' https://www.amd.com/en/developer/sev.html, 2016.

\bibitem{DocumentationX86Hygonsecurevirtualizationrst}
\BIBentryALTinterwordspacing
``Hygon secure virtualization,'' 2023. [Online]. Available: \url{https://gitee.com/anolis/cloud-kernel/blob/devel-5.10/Documentation/x86/hygon-secure-virtualization.rst}
\BIBentrySTDinterwordspacing

\bibitem{ju-shimAzureConfidentialComputing}
``Azure confidential computing overview,'' https://learn.microsoft.com/en-us/azure/confidential-computing/overview, 2024.

\bibitem{AWSNitroEnclaves}
``Nitro enclaves,'' https://aws.amazon.com/ec2/nitro/nitro-enclaves/, 2020.

\bibitem{GoogleConfidentialComputing}
\BIBentryALTinterwordspacing
``Google confidential computing,'' 2024. [Online]. Available: \url{https://cloud.google.com/security/products/confidential-computing}
\BIBentrySTDinterwordspacing

\bibitem{DBLP:conf/sp/SchusterCFGPMR15}
F.~Schuster, M.~Costa, C.~Fournet, C.~Gkantsidis, M.~Peinado, G.~Mainar{-}Ruiz, and M.~Russinovich, ``{VC3:} trustworthy data analytics in the cloud using {SGX},'' in \emph{Proc. of {IEEE} {S\&P}}, 2015.

\bibitem{DBLP:conf/sp/PriebeVC18}
C.~Priebe, K.~Vaswani, and M.~Costa, ``Enclavedb: {A} secure database using {SGX},'' in \emph{Proc. of {IEEE} {S\&P}}, 2018.

\bibitem{ohrimenkoObliviousMultiPartyMachine2016}
O.~Ohrimenko, F.~Schuster, C.~Fournet, A.~Mehta, S.~Nowozin, K.~Vaswani, and M.~Costa, ``Oblivious multi-party machine learning on trusted processors,'' in \emph{Proc. of {USENIX} Security}, 2016.

\bibitem{DBLP:conf/uss/HerwigGL20}
S.~Herwig, C.~Garman, and D.~Levin, ``Achieving keyless cdns with conclaves,'' in \emph{Proc. of {USENIX} Security}, 2020.

\bibitem{DBLP:conf/ccs/FischVBG17}
B.~Fisch, D.~Vinayagamurthy, D.~Boneh, and S.~Gorbunov, ``{IRON:} functional encryption using intel {SGX},'' in \emph{Proc. of {ACM} {CCS}}, 2017.

\bibitem{chengEkidenPlatformConfidentialitypreserving2019}
R.~Cheng, F.~Zhang, J.~Kos, W.~He, N.~Hynes, N.~M. Johnson, A.~Juels, A.~Miller, and D.~Song, ``Ekiden: {A} platform for confidentiality-preserving, trustworthy, and performant smart contracts,'' in \emph{Proc. of {IEEE} EuroS{\&}P}, 2019.

\bibitem{lindTeechainSecurePayment2019}
J.~Lind, O.~Naor, I.~Eyal, F.~Kelbert, E.~G. Sirer, and P.~R. Pietzuch, ``Teechain: a secure payment network with asynchronous blockchain access,'' in \emph{Proc. of {ACM} {SOSP}}, 2019.

\bibitem{zhangTownCrierAuthenticated2016}
F.~Zhang, E.~Cecchetti, K.~Croman, A.~Juels, and E.~Shi, ``Town crier: An authenticated data feed for smart contracts,'' in \emph{Proc. of {ACM} {CCS}}, 2016.

\bibitem{mateticBITEBitcoinLightweight2019}
S.~Matetic, K.~W{\"{u}}st, M.~Schneider, K.~Kostiainen, G.~Karame, and S.~Capkun, ``{BITE:} bitcoin lightweight client privacy using trusted execution,'' in \emph{Proc. of {USENIX} Security}, 2019.

\bibitem{mecury}
X.~Wen, Q.~Feng, J.~Niu, Y.~Zhang, and C.~Feng, ``Mercury: Practical cross-chain exchange via trusted hardware,'' \emph{{IEEE} Trans. Dependable Secur. Comput.}, vol.~23, no.~2, pp. 2949--2961, 2026.

\bibitem{TeeRollup}
X.~Wen, Q.~Feng, H.~Lyu, J.~Niu, Y.~Zhang, and C.~Feng, ``Teerollup: Efficient rollup design using heterogeneous {TEE},'' \emph{{IEEE} Trans. Computers}, vol.~74, no.~10, pp. 3546--3558, 2025.

\bibitem{chenSgxPectreStealingIntel2019}
G.~Chen, S.~Chen, Y.~Xiao, Y.~Zhang, Z.~Lin, and T.~Lai, ``Sgxpectre: Stealing intel secrets from {SGX} enclaves via speculative execution,'' in \emph{Proc. of {IEEE} EuroS{\&}P}, 2019.

\bibitem{bulckForeshadowExtractingKeys2018}
J.~V. Bulck, M.~Minkin, O.~Weisse, D.~Genkin, B.~Kasikci, F.~Piessens, M.~Silberstein, T.~F. Wenisch, Y.~Yarom, and R.~Strackx, ``Foreshadow: Extracting the keys to the intel {SGX} kingdom with transient out-of-order execution,'' in \emph{Proc. of {USENIX} Security}, 2018.

\bibitem{liSystematicLookCiphertext2022}
M.~Li, L.~Wilke, J.~Wichelmann, T.~Eisenbarth, R.~Teodorescu, and Y.~Zhang, ``A systematic look at ciphertext side channels on {AMD} {SEV-SNP},'' in \emph{Proc. {IEEE} {S\&P}}, 2022.

\bibitem{buhrenInsecureProvenUpdated2019}
R.~Buhren, C.~Werling, and J.~Seifert, ``Insecure until proven updated: Analyzing {AMD} sev's remote attestation,'' in \emph{Proc. of {ACM} {CCS}}, 2019.

\bibitem{buhrenOneGlitchRule2021}
R.~Buhren, H.~N. Jacob, T.~Krachenfels, and J.~Seifert, ``One glitch to rule them all: Fault injection attacks against amd's secure encrypted virtualization,'' in \emph{Proc. of {ACM} {CCS}}, 2021.

\bibitem{DBLP:conf/uss/LippGSMM16}
M.~Lipp, D.~Gruss, R.~Spreitzer, C.~Maurice, and S.~Mangard, ``Armageddon: Cache attacks on mobile devices,'' in \emph{Proc. of {USENIX} Security}, 2016.

\bibitem{hahnelHighResolutionSideChannels2017}
M.~H{\"{a}}hnel, W.~Cui, and M.~Peinado, ``High-resolution side channels for untrusted operating systems,'' in \emph{Proc. of {USENIX} {ATC}}, 2017.

\bibitem{leeInferringFinegrainedControl2017}
S.~Lee, M.~Shih, P.~Gera, T.~Kim, H.~Kim, and M.~Peinado, ``Inferring fine-grained control flow inside {SGX} enclaves with branch shadowing,'' in \emph{Proc. of {USENIX} Security}, 2017.

\bibitem{leeHackingDarknessReturnoriented2017}
J.~Lee, J.~S. Jang, Y.~Jang, N.~Kwak, Y.~Choi, C.~Choi, T.~Kim, M.~Peinado, and B.~B. Kang, ``Hacking in darkness: Return-oriented programming against secure enclaves,'' in \emph{Proc. of {USENIX} Security}, 2017.

\bibitem{bulckTellingYourSecrets2017}
J.~V. Bulck, N.~Weichbrodt, R.~Kapitza, F.~Piessens, and R.~Strackx, ``Telling your secrets without page faults: Stealthy page table-based attacks on enclaved execution,'' in \emph{Proc. of {USENIX} Security}, 2017.

\bibitem{DBLP:conf/ccs/WilkeS024}
L.~Wilke, F.~Sieck, and T.~Eisenbarth, ``Tdxdown: Single-stepping and instruction counting attacks against intel {TDX},'' in \emph{Proc. of {ACM} {CCS}}, 2024.

\bibitem{DBLP:conf/uss/RauscherWW0G25}
F.~Rauscher, L.~Wilke, H.~Weissteiner, T.~Eisenbarth, and D.~Gruss, ``Tdxploit: Novel techniques for single-stepping and cache attacks on intel {TDX},'' in \emph{Proc. of {USENIX} Security}, 2025.

\bibitem{emmadautermanReflectionsTrustingDistributed2022}
E.~Dauterman, V.~Fang, N.~Crooks, and R.~A. Popa, ``Reflections on trusting distributed trust,'' in \emph{Proc. of HotNets}, 2022.

\bibitem{kocherSpectreAttacksExploiting2019}
P.~Kocher, J.~Horn, A.~Fogh, D.~Genkin, D.~Gruss, W.~Haas, M.~Hamburg, M.~Lipp, S.~Mangard, T.~Prescher, M.~Schwarz, and Y.~Yarom, ``Spectre attacks: Exploiting speculative execution,'' in \emph{Proc. of {IEEE} {S\&P}}, 2019.

\bibitem{lippMeltdownReadingKernel2018}
M.~Lipp, M.~Schwarz, D.~Gruss, T.~Prescher, W.~Haas, A.~Fogh, J.~Horn, S.~Mangard, P.~Kocher, D.~Genkin, Y.~Yarom, and M.~Hamburg, ``Meltdown: Reading kernel memory from user space,'' in \emph{Proc. of {USENIX} Security}, 2018.

\bibitem{jauernigTrustedExecutionEnvironments2020}
P.~Jauernig, A.~Sadeghi, and E.~Stapf, ``Trusted execution environments: Properties, applications, and challenges,'' \emph{{IEEE} Secur. Priv.}, vol.~18, no.~2, pp. 56--60, 2020.

\bibitem{liSoKUnderstandingDesign2024}
M.~Li, Y.~Yang, G.~Chen, M.~Yan, and Y.~Zhang, ``Sok: Understanding design choices and pitfalls of trusted execution environments,'' in \emph{Proc. of the 19th {ACM} {AsiaCCS}}, 2024.

\bibitem{woodDenialServiceSensor2002a}
A.~D. Wood and J.~A. Stankovic, ``Denial of service in sensor networks,'' \emph{Computer}, vol.~35, no.~10, pp. 54--62, 2002.

\bibitem{zhouSideChannelAttacksTen2005}
Y.~Zhou and D.~Feng, ``Side-channel attacks: Ten years after its publication and the impacts on cryptographic module security testing,'' \emph{{IACR} Cryptol. ePrint Arch.}, p. 388, 2005.

\bibitem{connellSecretKeyRecovery2024}
G.~Connell, V.~Fang, R.~Schmidt, E.~Dauterman, and R.~A. Popa, ``Secret key recovery in a global-scale end-to-end encryption system,'' in \emph{Proc. of {USENIX} {OSDI}}, 2024.

\bibitem{SignalappSecureValueRecovery22025}
\BIBentryALTinterwordspacing
``Signalapp/securevaluerecovery2,'' 2025. [Online]. Available: \url{https://github.com/signalapp/SecureValueRecovery2}
\BIBentrySTDinterwordspacing

\bibitem{DBLP:conf/ccs/BentovJ0BDJ19}
I.~Bentov, Y.~Ji, F.~Zhang, L.~Breidenbach, P.~Daian, and A.~Juels, ``Tesseract: Real-time cryptocurrency exchange using trusted hardware,'' in \emph{Proc. of {ACM} {CCS}}, 2019.

\bibitem{herzbergProactiveSecretSharing1995}
A.~Herzberg, S.~Jarecki, H.~Krawczyk, and M.~Yung, ``Proactive secret sharing or: How to cope with perpetual leakage,'' in \emph{Proc. of {CRYPTO}}, 1995.

\bibitem{canettiAdaptiveSecurityThreshold1999}
R.~Canetti, R.~Gennaro, S.~Jarecki, H.~Krawczyk, and T.~Rabin, ``Adaptive security for threshold cryptosystems,'' in \emph{Proc. of {CRYPTO}}, 1999.

\bibitem{vassantlalCOBRADynamicProactive2022}
R.~Vassantlal, E.~Alchieri, B.~Ferreira, and A.~Bessani, ``{COBRA:} dynamic proactive secret sharing for confidential {BFT} services,'' in \emph{Proc. of {IEEE} {S\&P}}, 2022.

\bibitem{buterinNextgenerationSmartContract2014}
\BIBentryALTinterwordspacing
V.~Buterin, ``A next-generation smart contract and decentralized application platform,'' vol.~3, no.~37, pp. 2--1, 2014. [Online]. Available: \url{https://cryptorating.eu/whitepapers/Ethereum/Ethereum_white_paper.pdf}
\BIBentrySTDinterwordspacing

\bibitem{defilippiGovernanceBlockchainSystems2018}
P.~De~Filippi and G.~McMullen, ``Governance of blockchain systems: Governance of and by distributed infrastructure,'' Ph.D. dissertation, 2018.

\bibitem{bailleuAvocadoSecureInMemory2021}
M.~Bailleu, D.~Giantsidi, V.~Gavrielatos, D.~L. Quoc, V.~Nagarajan, and P.~Bhatotia, ``Avocado: {A} secure in-memory distributed storage system,'' in \emph{Proc. of {USENIX} {ATC}}, 2021.

\bibitem{pattukBigsecretSecureData2013}
E.~Pattuk, M.~Kantarcioglu, V.~Khadilkar, H.~Ulusoy, and S.~Mehrotra, ``Bigsecret: {A} secure data management framework for key-value stores,'' in \emph{Proc. of {IEEE} {CLOUD}}, 2013.

\bibitem{yuanBuildingEncryptedDistributed2016}
X.~Yuan, X.~Wang, C.~Wang, C.~Qian, and J.~Lin, ``Building an encrypted, distributed, and searchable key-value store,'' in \emph{Proc. of {ACM} AsiaCCS}, 2016.

\bibitem{mangipudiUncoveringImpactMental2023}
E.~V. Mangipudi, U.~Desai, M.~Minaei, M.~Mondal, and A.~Kate, ``Uncovering impact of mental models towards adoption of multi-device crypto-wallets,'' in \emph{Proc. of {ACM} {CCS}}, 2023.

\bibitem{yuDontPutAll2024}
Y.~Yu, T.~Sharma, S.~Das, and Y.~Wang, ``"don't put all your eggs in one basket": How cryptocurrency users choose and secure their wallets,'' in \emph{Proc. of {CHI}}, 2024.

\bibitem{liagourisSECRECYSecureCollaborative2023}
J.~Liagouris, V.~Kalavri, M.~Faisal, and M.~Varia, ``{SECRECY:} secure collaborative analytics in untrusted clouds,'' in \emph{Proc. of {USENIX} {NSDI}}, 2023.

\bibitem{corrigan-gibbsPrioPrivateRobust2017}
H.~Corrigan{-}Gibbs and D.~Boneh, ``Prio: Private, robust, and scalable computation of aggregate statistics,'' in \emph{Proc. of {USENIX} {NSDI}}, 2017.

\bibitem{volgushevDEMOIntegratingMPC2016}
N.~Volgushev, M.~Schwarzkopf, A.~Lapets, M.~Varia, and A.~Bestavros, ``{DEMO:} integrating {MPC} in big data workflows,'' in \emph{Proc. of {ACM} {CCS}}, 2016.

\bibitem{zhouShortcutMakingMPCbased2024}
P.~Zhou, X.~Guo, P.~Chen, T.~Li, S.~Lv, and Z.~Liu, ``Shortcut: Making mpc-based collaborative analytics efficient on dynamic databases,'' in \emph{Proc. of {ACM} {CCS}}, 2024.

\bibitem{yinHotStuffBFTConsensus2019}
M.~Yin, D.~Malkhi, M.~K. Reiter, G.~Golan{-}Gueta, and I.~Abraham, ``Hotstuff: {BFT} consensus with linearity and responsiveness,'' in \emph{Proc. of {ACM} {PODC}}, 2019.

\bibitem{bonehShortSignaturesWeil2001}
D.~Boneh, B.~Lynn, and H.~Shacham, ``Short signatures from the weil pairing,'' in \emph{Proc. of {ASIACRYPT}}, 2001.

\bibitem{kellerMPSPDZVersatileFramework2020}
M.~Keller, ``{MP-SPDZ:} {A} versatile framework for multi-party computation,'' in \emph{Proc. of {ACM} {CCS}}, 2020.

\bibitem{huntConfidentialComputingOpenPOWER2021}
G.~D.~H. Hunt, R.~Pai, M.~V. Le, H.~Jamjoom, S.~Bhattiprolu, R.~Boivie, L.~Dufour, B.~Frey, M.~Kapur, K.~A. Goldman, R.~Grimm, J.~Janakirman, J.~M. Ludden, P.~Mackerras, C.~May, E.~R. Palmer, B.~B. Rao, L.~Roy, W.~A. Starke, J.~Stuecheli, E.~Valdez, and W.~Voigt, ``Confidential computing for openpower,'' in \emph{Proc. of EuroSys}, 2021.

\bibitem{leeKeystoneOpenFramework2020}
D.~Lee, D.~Kohlbrenner, S.~Shinde, K.~Asanovic, and D.~Song, ``Keystone: an open framework for architecting trusted execution environments,'' in \emph{Proc. of EuroSys}, 2020.

\bibitem{costanSanctumMinimalHardware2016}
V.~Costan, I.~A. Lebedev, and S.~Devadas, ``Sanctum: Minimal hardware extensions for strong software isolation,'' in \emph{Proc. of {USENIX} Security}, 2016.

\bibitem{bahmaniCURESecurityArchitecture2021}
R.~Bahmani, F.~Brasser, G.~Dessouky, P.~Jauernig, M.~Klimmek, A.~Sadeghi, and E.~Stapf, ``{CURE:} {A} security architecture with customizable and resilient enclaves,'' in \emph{Proc. of {USENIX} Security}, 2021.

\bibitem{weiserTimbervTagisolatedMemory2019}
S.~Weiser, M.~Werner, F.~Brasser, M.~Malenko, S.~Mangard, and A.~Sadeghi, ``{TIMBER-V:} tag-isolated memory bringing fine-grained enclaves to {RISC-V},'' in \emph{Proc. of {NDSS}}, 2019.

\bibitem{xuControlledChannelAttacksDeterministic2015}
Y.~Xu, W.~Cui, and M.~Peinado, ``Controlled-channel attacks: Deterministic side channels for untrusted operating systems,'' in \emph{Proc. of {IEEE} {S\&P}}, 2015.

\bibitem{DBLP:conf/sp/MurdockOGBGP20}
K.~Murdock, D.~F. Oswald, F.~D. Garcia, J.~V. Bulck, D.~Gruss, and F.~Piessens, ``Plundervolt: Software-based fault injection attacks against intel {SGX},'' in \emph{Proc. of {IEEE} {S\&P}}, 2020.

\bibitem{DBLP:conf/uss/BorrelloKSLG022}
P.~Borrello, A.~Kogler, M.~Schwarzl, M.~Lipp, D.~Gruss, and M.~Schwarz, ``{\AE}pic leak: Architecturally leaking uninitialized data from the microarchitecture,'' in \emph{Proc. of {USENIX} Security}, 2022.

\bibitem{checkowayIagoAttacksWhy2013a}
S.~Checkoway and H.~Shacham, ``Iago attacks: why the system call {API} is a bad untrusted {RPC} interface,'' in \emph{Proc. of {ACM} {ASPLOS}}, 2013.

\bibitem{weichbrodtAsyncShockExploitingSynchronisation2016}
N.~Weichbrodt, A.~Kurmus, P.~R. Pietzuch, and R.~Kapitza, ``Asyncshock: Exploiting synchronisation bugs in intel {SGX} enclaves,'' in \emph{Proc. of {ESORICS}}, 2016.

\bibitem{shihTSGXEradicatingControlledChannel2017}
M.~Shih, S.~Lee, T.~Kim, and M.~Peinado, ``{T-SGX:} eradicating controlled-channel attacks against enclave programs,'' in \emph{Proc. of {NDSS}}, 2017.

\bibitem{wangLeakyCauldronDark2017}
W.~Wang, G.~Chen, X.~Pan, Y.~Zhang, X.~Wang, V.~Bindschaedler, H.~Tang, and C.~A. Gunter, ``Leaky cauldron on the dark land: Understanding memory side-channel hazards in {SGX},'' in \emph{Proc. of {ACM} {CCS}}, 2017.

\bibitem{shamirHowShareSecret1979}
A.~Shamir, ``How to share a secret,'' \emph{Commun. {ACM}}, vol.~22, no.~11, pp. 612--613, 1979.

\bibitem{desmedtSocietyGroupOriented1987}
Y.~Desmedt, ``Society and group oriented cryptography: {A} new concept,'' in \emph{Proc. of {CRYPTO}}, 1987.

\bibitem{zhouAPSSProactiveSecret2005}
L.~Zhou, F.~B. Schneider, and R.~van Renesse, ``{APSS:} proactive secret sharing in asynchronous systems,'' \emph{{ACM} Trans. Inf. Syst. Secur.}, vol.~8, no.~3, pp. 259--286, 2005.

\bibitem{schultzMPSSMobileProactive2010}
D.~A. Schultz, B.~Liskov, and M.~D. Liskov, ``{MPSS:} mobile proactive secret sharing,'' \emph{{ACM} Trans. Inf. Syst. Secur.}, vol.~13, no.~4, pp. 34:1--34:32, 2010.

\bibitem{baronCommunicationoptimalProactiveSecret2015}
J.~Baron, K.~E. Defrawy, J.~Lampkins, and R.~Ostrovsky, ``Communication-optimal proactive secret sharing for dynamic groups,'' in \emph{Proc. of {ACNS}}, 2015.

\bibitem{maramCHURPDynamiccommitteeProactive2019}
S.~K.~D. Maram, F.~Zhang, L.~Wang, A.~Low, Y.~Zhang, A.~Juels, and D.~Song, ``{CHURP:} dynamic-committee proactive secret sharing,'' in \emph{Proc. of {ACM} {CCS}}, 2019.

\bibitem{yaoProtocolsSecureComputations1982}
A.~C. Yao, ``Protocols for secure computations (extended abstract),'' in \emph{Proc. of {FOCS}}, 1982.

\bibitem{lindellSecureMultipartyComputation2021}
Y.~Lindell, ``Secure multiparty computation,'' \emph{Commun. {ACM}}, vol.~64, no.~1, pp. 86--96, 2021.

\bibitem{escuderoIntroductionSecretSharingBasedSecure2022}
D.~Escudero, ``An introduction to secret-sharing-based secure multiparty computation,'' \emph{{IACR} Cryptol. ePrint Arch.}, p.~62, 2022.

\bibitem{castroPracticalByzantineFault2002}
M.~Castro and B.~Liskov, ``Practical byzantine fault tolerance and proactive recovery,'' \emph{{ACM} Trans. Comput. Syst.}, vol.~20, no.~4, pp. 398--461, 2002.

\bibitem{TrustedComputingBase}
\BIBentryALTinterwordspacing
``Trusted computing base recovery,'' Intel, 2025. [Online]. Available: \url{https://www.intel.com/content/www/us/en/developer/articles/technical/software-security-guidance/best-practices/trusted-computing-base-recovery.html}
\BIBentrySTDinterwordspacing

\bibitem{AffectedProcessorsTransient}
\BIBentryALTinterwordspacing
``Affected processors: Transient execution attacks \& related security,'' Intel, 2025. [Online]. Available: \url{https://www.intel.com/content/www/us/en/developer/topic-technology/software-security-guidance/processors-affected-consolidated-product-cpu-model.html}
\BIBentrySTDinterwordspacing

\bibitem{AMDSEVConfidential}
\BIBentryALTinterwordspacing
``Amd sev confidential computing vulnerability,'' AMD, 2025. [Online]. Available: \url{https://www.amd.com/en/resources/product-security/bulletin/amd-sb-3019.html}
\BIBentrySTDinterwordspacing

\bibitem{Ladon2025}
H.~Lyu, S.~Xie, J.~Niu, C.~Feng, Y.~Zhang, and I.~Beschastnikh, ``Ladon: High-performance multi-bft consensus via dynamic global ordering,'' in \emph{Proc. of EuroSys}, 2025.

\bibitem{fast-hotstuff}
M.~M. Jalalzai, J.~Niu, C.~Feng, and F.~Gai, ``Fast-hotstuff: {A} fast and robust {BFT} protocol for blockchains,'' \emph{{IEEE} Trans. Dependable Secur. Comput.}, vol.~21, no.~4, pp. 2478--2493, 2024.

\bibitem{feldmanPracticalSchemeNoninteractive1987}
P.~Feldman, ``A practical scheme for non-interactive verifiable secret sharing,'' in \emph{Proc. of {FOCS}}, 1987.

\bibitem{chenMAGEMutualAttestation2022}
G.~Chen and Y.~Zhang, ``{MAGE:} mutual attestation for a group of enclaves without trusted third parties,'' in \emph{Proc. {USENIX} Security}, 2022.

\bibitem{lindellFrameworkConstructingFast2017a}
Y.~Lindell and A.~Nof, ``A framework for constructing fast {MPC} over arithmetic circuits with malicious adversaries and an honest-majority,'' in \emph{Proc. of {ACM} {CCS}}, 2017.

\bibitem{HotStuffLibhotstuff2025}
``Hot-stuff/libhotstuff,'' \url{https://github.com/hot-stuff/libhotstuff}, Jan. 2018.

\bibitem{IntelTDXConnect}
``Intel{\textregistered} tdx connect architecture specification,'' https://www.intel.com/content/www/us/en/content-details/773614/intel-tdx-connect-architecture-specification.html, 2023.

\bibitem{sousaByzantineConsensusBFT2012}
J.~Sousa and A.~N. Bessani, ``From byzantine consensus to {BFT} state machine replication: {A} latency-optimal transformation,'' in \emph{Proc. of IEEE EDCC}, 2012.

\bibitem{AWSKeyManagement}
\BIBentryALTinterwordspacing
(2026) Aws key management service documentation. [Online]. Available: \url{https://docs.aws.amazon.com/kms/}
\BIBentrySTDinterwordspacing

\bibitem{msmbaldwinAzureKeyVaulta}
\BIBentryALTinterwordspacing
(2026) Azure key vault overview - azure key vault. [Online]. Available: \url{https://learn.microsoft.com/en-us/azure/key-vault/general/overview}
\BIBentrySTDinterwordspacing

\bibitem{SafeWallet}
\BIBentryALTinterwordspacing
(2026) Safe\{Wallet\}. [Online]. Available: \url{https://safe.global/}
\BIBentrySTDinterwordspacing

\bibitem{fireblocks}
\BIBentryALTinterwordspacing
(2026) Fireblocks. [Online]. Available: \url{https://www.fireblocks.com/}
\BIBentrySTDinterwordspacing

\bibitem{DBLP:conf/asiacrypt/LepointPRST21}
T.~Lepoint, S.~Patel, M.~Raykova, K.~Seth, and N.~Trieu, ``Private join and compute from {PIR} with default,'' in \emph{Proc. of {ASIACRYPT}}, 2021.

\bibitem{mateticROTERollbackProtection2017}
S.~Matetic, M.~Ahmed, K.~Kostiainen, A.~Dhar, D.~M. Sommer, A.~Gervais, A.~Juels, and S.~Capkun, ``{ROTE:} rollback protection for trusted execution,'' in \emph{Proc. of {USENIX} Security}, 2017.

\bibitem{niuNARRATORSecurePractical2022}
J.~Niu, W.~Peng, X.~Zhang, and Y.~Zhang, ``{NARRATOR:} secure and practical state continuity for trusted execution in the cloud,'' in \emph{Proc. of {ACM} {CCS}}, 2022.

\bibitem{narrator-pro}
W.~Peng, X.~Li, J.~Niu, X.~Zhang, and Y.~Zhang, ``Ensuring state continuity for confidential computing: {A} blockchain-based approach,'' \emph{{IEEE} Trans. Dependable Secur. Comput.}, vol.~21, no.~6, pp. 5635--5649, 2024.

\bibitem{tiks}
W.~Wang, J.~Niu, M.~K. Reiter, and Y.~Zhang, ``Formally verifying a rollback-prevention protocol for tees,'' in \emph{Proc. of FORTE}, 2024.

\bibitem{angelNimbleRollbackProtection2023}
S.~Angel, A.~Basu, W.~Cui, T.~Jaeger, S.~Lau, S.~T.~V. Setty, and S.~Singanamalla, ``Nimble: Rollback protection for confidential cloud services,'' in \emph{Proc. of {USENIX} {OSDI}}, 2023.

\bibitem{dinisRRFaultModel2023}
B.~Dinis, P.~Druschel, and R.~Rodrigues, ``{RR:} {A} fault model for efficient {TEE} replication,'' in \emph{Proc. of {NDSS}}, 2023.

\bibitem{Achilles}
J.~Niu, X.~Wen, G.~Wu, S.~Liu, J.~Yu, and Y.~Zhang, ``Achilles: Efficient tee-assisted {BFT} consensus via rollback resilient recovery,'' in \emph{Proc. of EuroSys}, 2025.

\bibitem{wangENGRAFTEnclaveguardedRaft2022}
W.~Wang, S.~Deng, J.~Niu, M.~K. Reiter, and Y.~Zhang, ``{ENGRAFT:} enclave-guarded raft on byzantine faulty nodes,'' in \emph{Proc. of {ACM} {CCS}}, 2022.

\bibitem{DBLP:journals/pvldb/HowardAACCCDFJK23}
H.~Howard, F.~Alder, E.~Ashton, A.~Chamayou, S.~Clebsch, M.~Costa, A.~Delignat{-}Lavaud, C.~Fournet, A.~Jeffery, M.~Kerner, F.~Kounelis, M.~A. Kuppe, J.~Maffre, M.~Russinovich, and C.~M. Wintersteiger, ``Confidential consortium framework: Secure multiparty applications with confidentiality, integrity, and high availability,'' \emph{Proc. {VLDB} Endow.}, vol.~17, no.~2, pp. 225--240, 2023.

\bibitem{riazi2018chameleon}
M.~S. Riazi, C.~Weinert, O.~Tkachenko, E.~M. Songhori, T.~Schneider, and F.~Koushanfar, ``Chameleon: {A} hybrid secure computation framework for machine learning applications,'' in \emph{Proc. of {ACM} {AsiaCCS}}, 2018.

\bibitem{luCorrelatedRandomnessTeleportation2021}
Y.~Lu, B.~Zhang, H.~Zhou, W.~Liu, L.~Zhang, and K.~Ren, ``Correlated randomness teleportation via semi-trusted hardware - enabling silent multi-party computation,'' in \emph{Proc. of {ESORICS}}, 2021.

\bibitem{wuHybridTrustMultiparty2022}
P.~Wu, J.~Ning, J.~Shen, H.~Wang, and E.~Chang, ``Hybrid trust multi-party computation with trusted execution environment,'' in \emph{Proc. of {NDSS}}, 2022.

\bibitem{huEfficientPracticalMultiparty2025}
X.~Hu, R.~Li, Y.~Liu, and Q.~Wang, ``Towards efficient and practical multi-party computation under inconsistent trust in tees,'' in \emph{Proc. of {IEEE} {S\&P}}, 2025.

\bibitem{DeploymentDilemmaMerits2023}
\BIBentryALTinterwordspacing
``The deployment dilemma: Merits \& challenges of deploying mpc,'' 2023. [Online]. Available: \url{https://mpc.cs.berkeley.edu/blog/deployment-dilemma.html}
\BIBentrySTDinterwordspacing

\bibitem{goyalUnconditionalCommunicationEfficientMPC2021}
V.~Goyal, A.~Polychroniadou, and Y.~Song, ``Unconditional communication-efficient {MPC} via hall's marriage theorem,'' in \emph{Proc. of {CRYPTO}}, 2021.

\bibitem{beckScalableMultipartyGarbling2023}
G.~Beck, A.~Goel, A.~Hegde, A.~Jain, Z.~Jin, and G.~Kaptchuk, ``Scalable multiparty garbling,'' in \emph{Proc. of {ACM} {CCS}}, 2023.

\bibitem{googleconfidentialspace}
\BIBentryALTinterwordspacing
(2023) How confidential space and mpc can help secure digital assets. [Online]. Available: \url{https://cloud.google.com/blog/products/identity-security/how-confidential-space-and-mpc-can-help-secure-digital-assets}
\BIBentrySTDinterwordspacing

\bibitem{DBLP:conf/uss/VitiSGD0BH0D25}
R.~D. Viti, I.~Sheff, N.~Glaeser, B.~Dinis, R.~Rodrigues, B.~Bhattacharjee, A.~Hithnawi, D.~Garg, and P.~Druschel, ``Covault: Secure, scalable analytics of personal data,'' in \emph{Proc. of {USENIX} Security}, 2025.

\end{thebibliography}

\appendix
\balance
\subsection{Complementarity of MPC and Heterogeneous TEEs} \label{appen:mpc}
Deploying MPC in real systems faces a well-known deployment dilemma between security and performance~\cite{DeploymentDilemmaMerits2023}. On the one hand, MPC requires large and diverse committees to achieve strong security; small or closely related parties undermine the independence that secret sharing relies on. On the other hand, MPC communication and coordination costs grow with both the number of parties and the number of interactive multiplications, making large committees expensive in bandwidth and latency~\cite{goyalUnconditionalCommunicationEfficientMPC2021, beckScalableMultipartyGarbling2023}.

MPC and TEEs address different parts of the problem. MPC protects secrets by distributing them across parties and computing over shares, but it does not decide which machines should remain eligible to hold shares as platforms are patched, revoked, or replaced. TEEs provide attested execution and platform-state evidence, but a compromised enclave may expose its local plaintext state, so TEE protection alone does not invalidate previously leaked shares. \sysname combines these mechanisms at the system layer: heterogeneous TEEs provide concrete trust domains and attestation-derived eligibility, while DPSS/MPC maintains confidentiality and computation over distributed shares. When evidence changes the committee, \sysname couples the certified configuration transition with refresh, resharing, or recovery, preventing past compromise from accumulating across configurations. Previously related designs already appear in secure recovery, digital-asset custody, threshold signing/wallet infrastructure, and privacy-preserving analytics (e.g., Signal SVR3~\cite{connellSecretKeyRecovery2024}, Confidential Space + MPC~\cite{googleconfidentialspace}, Fireblocks~\cite{fireblocks}, CoVault~\cite{DBLP:conf/uss/VitiSGD0BH0D25}).

\subsection{Performance of CSV}
\label{appen:csv_expriment}
To further evaluate the generality and performance of \sysname, we conducted additional experiments on the CSV platform. The experimental setup and evaluation methodology are consistent with those described in ~\secref{experiment_setup}.

\begin{table}[ht]
    \centering
    \setlength{\abovecaptionskip}{2pt}
    \renewcommand{\arraystretch}{0.95}
    \setlength{\tabcolsep}{4pt}
    \caption{Recovery and reshare latency (s) with 1K entries.}
    \label{tab:recovery_reshare_latency_csv}

    \begin{tabular*}{\columnwidth}{@{\extracolsep{\fill}}lccccc@{}}
        \toprule
        Item & Metric & $n=4$ & $n=10$ & $n=16$ & $n=25$ \\
        \midrule
        \multirow{2}{*}{CSV-Recovery}
            & Poly. Gen.   & 1.826 & 4.006 & 7.828  & 16.918 \\
            & State Recon. & 0.307 & 0.545 & 0.914  & 1.381 \\
        \midrule
        \multirow{2}{*}{CSV-Reshare}
            & Poly. Gen.   & 3.274 & 9.288 & 17.257 & 36.018 \\
            & State Recon. & 0.804 & 1.478 & 2.259  & 4.287 \\
        \bottomrule
    \end{tabular*}

    \vspace{-2mm}
\end{table}

\tabref{tab:recovery_reshare_latency_csv} reports the latency of \sysnameCSV during reconfiguration and fault node recovery protocols. The results are consistent with previous findings: at $n=4$, \sysnameCSV achieves lower latency than COBRA. However, as the number of nodes increases, the TEE overhead causes \sysnameCSV's latency to gradually surpass that of COBRA.

\tabref{tab:csv-performance} summarizes the maximum throughput and corresponding latency of \sysnameCSV. \sysnameCSV achieves higher throughput than COBRA for when $n < 25$, but its throughput decreases more rapidly as the number of nodes increases. At $n=25$, the write throughput of \sysnameCSV becomes comparable to COBRA.

\begin{table}[ht]
  \centering
  \caption{Performance of \sysnameCSV.}
  \label{tab:csv-performance}
  \begin{tabular}{l|ccc|ccc}
    \toprule
    Item &
    \multicolumn{3}{c|}{Throughput (kTPS)} &
    \multicolumn{3}{c}{Latency (ms)} \\
    &
    $n = 4$ & $n = 10$ & $n = 25$ &
    $n = 4$ & $n = 10$ & $n = 25$ \\
    \midrule
    Write & 22.5 & 15.6 & 8.4 & 64.1 & 97.9 & 190.8 \\
    Read & 120.1 & 40.5 & 9.1  & 3.6 & 5.1 & 11.0 \\
    \bottomrule
  \end{tabular}
\end{table}

\end{document}